\documentclass[aps,reprint,superscriptaddress,longbibliography]{revtex4-1}
\usepackage{amsmath}
\usepackage{amssymb}
\usepackage{graphicx,subfigure}
\usepackage[bookmarks=false]{hyperref}
\usepackage{bm,bbm}
\usepackage{xcolor}
\hypersetup{colorlinks=true,citecolor=blue,linkcolor=red,%
urlcolor=blue,pdfstartview=FitH,bookmarksopen=true}

\usepackage{amsthm}
\newtheoremstyle{note}
{3pt}
{3pt}
{}
{}
{\itshape}
{:}
{.5em}
{}

\usepackage[T1]{fontenc}
\usepackage{anyfontsize}
\usepackage[osf,sc]{mathpazo}
\usepackage{dsfont}

\newtheorem{theorem}{Theorem}
\newtheorem{proposition}[theorem]{Proposition}

\newtheorem{lemma}[theorem]{Lemma}
\newtheorem{corollary}[theorem]{Corollary}
\newtheorem{observation}[theorem]{Observation}

\DeclareMathOperator{\Tr}{Tr}

\DeclareMathOperator{\rank}{rank}
\DeclareMathOperator{\conv}{conv}

\newcommand{\bra}[1]{\langle #1\rvert}
\newcommand{\ket}[1]{\lvert #1\rangle}
\newcommand{\mean}[1]{\langle #1\rangle}

\newcommand{\braket}[2]{\langle #1\vert #2\rangle}
\newcommand{\ketbra}[1]{\ket{#1}\bra{#1}}
\newcommand{\abs}[1]{\lvert #1\rvert}
\newcommand{\norm}[1]{\lVert #1\rVert}
\newcommand{\vect}[1]{\bm{#1}}

\newcommand{\id}{\mathrm{id}}

\newcommand{\I}{\mathds{1}}
\newcommand{\FA}{\mathrm{\quad\forall~}}
\newcommand{\Sep}{\mathrm{SEP}}

\newcommand{\AME}{\mathrm{AME}}

\newcommand{\maxover}[1][]{\underset{#1}{\mathrm{max}}}
\newcommand{\minover}[1][]{\underset{#1}{\mathrm{min}}}
\newcommand{\findover}[1][]{\underset{#1}{\mathrm{find}}}
\newcommand{\subto}{\mathrm{~s.t.}}

\newcommand{\vl}{\bm{\lambda}}

\newcommand{\bI}{{I^c}}
\newcommand{\cC}{\mathcal{C}}

\newcommand{\cP}{\mathcal{P}}

\newcommand{\dC}{\mathds{C}}

\newcommand{\dR}{\mathds{R}}

\newcommand{\Appendix}{Appendix}
\newcommand{\Appendices}{Appendices}

\newcommand{\AEq}{Eq.}
\newcommand{\AEqs}{Eqs.}

\newcommand{\AEquation}{Equation}
\newcommand{\AEquations}{Equations}

\newcommand{\MEq}{Eq.}
\newcommand{\MEqs}{Eqs.}

\begin{document}

\title{A complete hierarchy for the pure state marginal problem in quantum 
mechanics}

\author{Xiao-Dong Yu}
\affiliation{Naturwissenschaftlich-Technische Fakult\"at, Universit\"at Siegen,
Walter-Flex-Str. 3, D-57068 Siegen, Germany}

\author{Timo Simnacher}
\affiliation{Naturwissenschaftlich-Technische Fakult\"at, Universit\"at Siegen,
Walter-Flex-Str. 3, D-57068 Siegen, Germany}

\author{Nikolai Wyderka}
\affiliation{Naturwissenschaftlich-Technische Fakult\"at, Universit\"at Siegen,
Walter-Flex-Str. 3, D-57068 Siegen, Germany}
\affiliation{Institut f\"ur Theoretische Physik III,
Heinrich-Heine-Universit\"at D\"usseldorf,
Universit\"atsstr. 1, D-40225 D\"usseldorf, Germany}

\author{H. Chau Nguyen}
\affiliation{Naturwissenschaftlich-Technische Fakult\"at, Universit\"at Siegen,
Walter-Flex-Str. 3, D-57068 Siegen, Germany}

\author{Otfried G\"uhne}
\affiliation{Naturwissenschaftlich-Technische Fakult\"at, Universit\"at Siegen,
Walter-Flex-Str. 3, D-57068 Siegen, Germany}

\date{\today}

\begin{abstract}
  Clarifying the relation between the whole and its parts is crucial for many 
  problems in science. In quantum mechanics, this question manifests itself in 
  the quantum marginal problem, which asks whether there is a global pure 
  quantum state for some given marginals. This problem arises in many contexts, 
  ranging from quantum chemistry to entanglement theory and quantum error 
  correcting codes. In this paper, we prove a correspondence of the marginal 
  problem to the separability problem.  Based on this, we describe a sequence 
  of semidefinite programs which can decide whether some given marginals are 
  compatible with some pure global quantum state. As an application, we prove 
  that the existence of multiparticle absolutely maximally entangled states for 
  a given dimension is equivalent to the separability of an explicitly given 
  two-party quantum state.  Finally, we show that the existence of quantum 
  codes with given parameters can also be interpreted as a marginal problem, 
  hence, our complete hierarchy can also be used.
\end{abstract}

\maketitle

\section{Introduction}
%
For a given multiparticle quantum state $\ket{\varphi}$
it is straightforward to compute its marginals or reduced 
density matrices on some subsets of the particles. The reverse 
question, whether a given set of marginals is compatible with 
a global pure state, is, however, not easy to decide. Still, 
it is at the heart of many problems in quantum physics. Already 
in the early days it was a key motivation for Schr\"odinger to study 
entanglement \cite{Schroedinger1935}, and it was recognized as a central 
problem in quantum chemistry \cite{Coleman1963}. There, often additional 
constraints play a role, e.g., if one considers fermionic systems.  Then, the 
anti-symmetry leads to additional constraints on the marginals, generalizing 
the Pauli principle \cite{Klyachko2006,Schilling2015}.  A variation of the 
marginal problem is the question whether or not the marginals determine the 
global state uniquely or not 
\cite{Linden.etal2002,Sawicki.etal2013,Wyderka.etal2017}.  This is relevant in 
condensed matter physics, where one may ask whether a state is the unique 
ground state of a local Hamiltonian \cite{Huber.Guehne2016,Karuvade.etal2019}.  
Many other cases, such as marginal problems for Gaussian and symmetric states 
\cite{Eisert.etal2008,Aloy.etal2020} and applications in quantum correlation 
\cite{Walter.etal2013}, quantum causality \cite{Chaves.etal2015}, and 
interacting quantum many-body systems 
\cite{Schilling.etal2020,Maciazek.etal2020} have been studied.

With the emergence of quantum information processing, various 
specifications of the marginal problem moved into the center 
of attention. In entanglement theory a pure two-particle state 
is maximally entangled, if the one-particle marginals are maximally 
mixed. Furthermore, absolutely maximally entangled (AME) states are multiparticle 
states that are maximally entangled for any bipartition. This makes them 
valuable ingredients for quantum information protocols 
\cite{Helwig.etal2012,Helwig.Cui2013}, but it turns out that AME states do not 
exist for arbitrary dimensions, as not always global states with the desired 
mixed marginals can be found 
\cite{Scott2004,Goyeneche.etal2015,Huber.etal2017,Huber.etal2018}. In fact, 
also states obeying weaker conditions, where a smaller number of marginals 
should be maximally mixed, are of fundamental interest, but in general it is 
open when such states exist \cite{Bryan.etal2019,Raissi.etal2020,Grassl2007}.  
More generally, the construction of quantum error correcting codes, which 
constitute fundamental building blocks in the design of quantum computer 
architectures \cite{Ladd.etal2010,Preskill2018,Arute.etal2019},
essentially amounts to the identification of subspaces of the total Hilbert 
space, where all states in this space obey certain marginal constraints. This 
establishes a connection to the AME problem, which consequently was announced 
to be one of the central problems in quantum information theory 
\cite{Horodecki.etal2020}.

\begin{figure}[t]
  \centering
  \includegraphics[width=.45\textwidth]{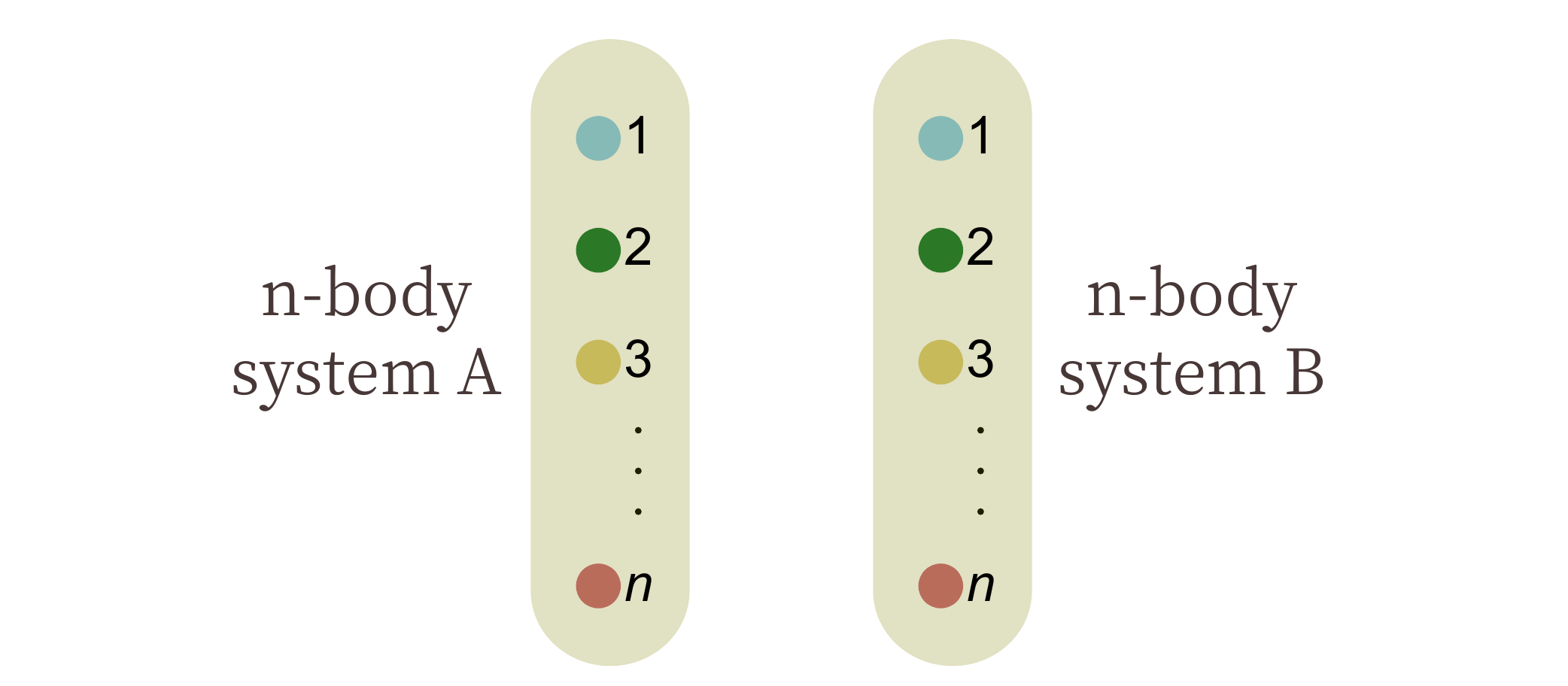}
  \caption{\textbf{Two-party extension for the marginal problem.} In the 
    marginal problem one aims to characterize the pure states $\ket{\varphi}$ 
    on $n$ particles, which are compatible with given marginals. The key idea 
    of our approach is to drop the purity constraint and to consider mixed 
    states $\rho$ with the given marginals. Then, the purity is enforced by 
  considering a two-party extension~$\Phi_{AB}$.}
  \label{fig:twocopy}
\end{figure}

In this paper, we rewrite the marginal problem as an optimization problem over 
separable states. Here and in the following, the term marginal problem usually 
refers to the pure state marginal problem in quantum mechanics. This rewriting 
allows us to transform the non-convex and thus intractable purity constraint 
into a complete hierarchy of conditions for a set of marginals to be compatible 
with a global pure state.  Each step is given by a semidefinite program (SDP), 
the conditions become stronger with each level, and a set of marginals comes 
from a global state, if and only if all steps are passed. There are at least 
two advantages of writing the marginal problem as an SDP hierarchy: First, the 
symmetry in the physical problem can be directly incorporated to drastically 
simplify the optimization (or feasibility) problem.  Second, many known 
efficient and reliable algorithms are known for solving SDPs 
\cite{Boyd.Vandenberghe2004}, which is in stark contrast to nonconvex 
optimization.
To show the effectiveness of our method, we consider the existence problem of 
AME states. By employing the symmetry, we show that an AME state for a given 
number of particles and dimension exists, if and only if a specific two-party 
quantum state is separable. In fact, this allows us to reproduce nearly all 
previous results on the AME problem \cite{Huber.Wyderka2018} with only few 
lines of calculation. Finally, we show that our approach can also be extended 
to study the existence problem of quantum codes.

\section{Results}
\textbf{Connecting the marginal problem with the separability problem.}
%
The formal definition of the marginal problem is the following: Consider
an $n$-particle Hilbert space $\mathcal{H}=\bigotimes_{i=1}^n\mathcal{H}_i$,
and let $\mathcal{I}\subset\{I\mid I\subset [n]=\{1,2,\dots,n\}\}$ be some
subsets of the particles, where the reduced states $\rho_I$  are known marginals.
Then, the problem reads
\begin{equation}
\begin{aligned}
&\findover \quad && \ket{\varphi}\\
&\subto && \Tr_{\bI}(\ket{\varphi}\bra{\varphi})=\rho_I,~I\in\mathcal{I}.
\end{aligned}
\label{eq:marginalPurity}
\end{equation}
Here, $\bI=[n]\setminus I$ denotes the complement of the set $I$. Before explaining our
approach, two facts are worth mentioning: First, if the global state $\ket{\varphi}\bra{\varphi}$ 
is not required to be pure, then the quantum marginal problem without purity 
constraint is already an SDP. Second, if the given marginals are only one-body 
marginals, that is $\mathcal{I}= \{ \{i\} \mid i \in [n] \}$, the marginals are 
non-overlapping and the problem in Eq.~\eqref{eq:marginalPurity} was solved by 
Klyachko \cite{Klyachko2004}. For overlapping marginals, however, the solution 
is more complicated, and this is what we want to discuss in this work.

The main idea of our method is to consider, for a given set of marginals, 
the compatible states and their extensions to two copies. Then, we can
formulate the purity constraint using an SDP. First, let us introduce some 
notation. Let $\cC$ be the set of global states (not necessarily pure) that 
are compatible with the marginals, i.e.,
\begin{equation}
\cC=\{\rho\mid\rho\ge 0,~\Tr_{\bI}(\rho)=\rho_I \FA I\in\mathcal{I}\}.
\label{eq:cC}
\end{equation}
Then, we define $\cC_2$ to be the convex hull of two copies of the 
compatible states
\begin{equation}
\cC_2=\conv\{\rho\otimes\rho\mid\rho\in\cC\}=
\Big\{\sum_\mu p_\mu\rho_\mu\otimes\rho_\mu\mid\rho_\mu\in\cC
\Big\},
\label{eq:cCN}
\end{equation}
where the  $p_\mu$ form a probability distribution.
We denote the two parties as $A$ and $B$, and each of them owns an $n$-body 
quantum system; see Fig.~\ref{fig:twocopy}. 

To impose the purity constraint, 
we take advantage of the well-known relation~\cite{Werner1989}
\begin{equation}
\Tr(V_{AB}\rho_A\otimes\rho_B)=\Tr(\rho_A\rho_B),
\label{eq:SWAP}
\end{equation}
where $\rho_A$ and $\rho_B$ are arbitrary quantum states, and $V_{AB}$ is the 
swap operator between parties $A$ and $B$, i.e.,
\begin{equation}
  V_{AB}=\sum_{i,j}\ket{j,i}\bra{i,j},
\end{equation}
which acts on a state $\Phi_{AB}=\sum_{i,j}\omega_{ijkl}\ket{i,j}\bra{k,l}$ as $V_{AB}\Phi_{AB}=\sum_{i,j}\omega_{ijkl}\ket{j,i}\bra{k,l}$.
For a state  $\Phi_{AB}$ in $\cC_2$  this implies that
\begin{equation}
  \Tr(V_{AB}\Phi_{AB})=\sum_\mu p_\mu\Tr(\rho_\mu^2)\le 1.
  \label{eq:purity}
\end{equation}
Furthermore, equality in Eq.~\eqref{eq:purity} is attained if
and only if all $\rho_\mu$ are pure states. This leads to our first
key observation: There exists a pure state in $\cC$ if and only if 
$\max_{\Phi_{AB}\in\cC_2}\Tr(V_{AB}\Phi_{AB})=1$.

What remains to be done is the characterization of the set 
$\cC_2$, then we can formulate the quantum marginal problem as an optimization 
problem over this set. This can be done by taking advantage of the separability 
property \cite{Werner1989} of the states in $\cC_2$ with respect to the 
bipartition $(A|B)$ .

\begin{theorem}
  There exists a pure quantum state $\ket{\varphi}$ that satisfies
  $\Tr_\bI(\ket{\varphi}\bra{\varphi})=\rho_I$ for all $I\in\mathcal{I}$ 
  if, and only if, the solution of the following convex optimization is equal 
  to one,
    \begin{align}
      &\maxover[\Phi_{AB}] && \Tr(V_{AB}\Phi_{AB})\label{eq:marginalPure1} \\
      &\subto && \Phi_{AB}\in\Sep,~\Tr(\Phi_{AB})=1,\label{eq:marginalPure2}\\
      &       && \Tr_{A_\bI B_\bI}(\Phi_{AB})=\rho_I\otimes\rho_I\FA 
      I\in\mathcal{I},\label{eq:marginalPure3}
    \end{align}
  where $\Sep$ denotes the set of separable states w.r.t.~the bipartition 
  $(A|B)$, $A_\bI$ denotes all subsystems $A_i$ for $i\in\bI$, and similarly 
  for $B_\bI$.
  \label{thm:marginalPure}
\end{theorem}

\begin{proof}
  On the one hand, if there exists a pure state $\ket{\varphi}\bra{\varphi}\in\cC$, 
  one can easily verify that 
  $\Phi_{AB}=\ket{\varphi}\bra{\varphi}\otimes\ket{\varphi}\bra{\varphi}$ 
  satisfies the constraints in 
  Eqs.~(\ref{eq:marginalPure2},\,\ref{eq:marginalPure3}) as well as 
  $\Tr(V_{AB}\Phi_{AB})=1$.

  On the other hand, if the solution of Eq.~\eqref{eq:marginalPure1} is equal 
  to one, then the separability constraint and Eq.~\eqref{eq:purity} imply that 
  $\Phi_{AB}$ must be of the form \cite{Toth.Guehne2009}
  \begin{equation}
    \Phi_{AB}=\sum_\mu p_\mu\ket{\psi_\mu}\bra{\psi_\mu}\otimes
    \ket{\psi_\mu}\bra{\psi_\mu}.
    \label{eq:valPhi}
  \end{equation}
  Writing $\Tr_\bI(\ket{\psi_\mu}\bra{\psi_\mu})=\sigma_I^{(\mu)}$,  
  the constraint in Eq.~\eqref{eq:marginalPure3} 
  implies that
  \begin{equation}
    \sum_\mu p_\mu\sigma_I^{(\mu)}\otimes\sigma_I^{(\mu)}=\rho_I\otimes\rho_I
    \FA I\in\mathcal{I}.
    \label{eq:extremeTensor}
  \end{equation}
  Without loss of generality, we assume that all $p_\mu$ are strictly positive 
  and we want to show that all $\sigma_I^{(\mu)}=\rho_I$, which will imply that 
  each $\ket{\psi_\mu}$ is a pure state with the desired marginals. Let $X$ be 
  any Hermitian matrix such that $\Tr(X\rho_I)=0$, then 
  Eq.~\eqref{eq:extremeTensor} and the relation $\Tr[(X\otimes 
  X)(\sigma\otimes\sigma)]=[\Tr(X\sigma)]^2$ imply that
  \begin{equation}
    \sum_{\mu}p_\mu\left[\Tr\left(X\sigma_I^{(\mu)}\right)\right]^2
    =[\Tr(X\rho_I)]^2
    =0.
    \label{eq:orthogonalRho}
  \end{equation}
  By noting that $\Tr(X\sigma_I^{(\mu)})$ are real numbers, we get that 
  $\Tr(X\sigma_I^{(\mu)})=0$, for all $\mu$ and all $X$ such that 
  $\Tr(X\rho_I)=0$. Thus, all $\sigma_I^{(\mu)}$ are proportional to $\rho_I$ 
  and, together with $\Tr(\rho_I)=\Tr(\sigma_I^{(\mu)})=1$, this implies that 
  $\sigma_I^{(\mu)}=\rho_I$ for all $\mu$.
\end{proof}


Before proceeding further, we would like to add a few remarks. First, in 
Theorem~\ref{thm:marginalPure} the constraint in Eq.~\eqref{eq:marginalPure3} 
can be replaced by a stronger condition
\begin{equation}
  \Tr_{A_\bI}(\Phi_{AB})=\rho_I\otimes\Tr_A(\Phi_{AB})\FA I\in\mathcal{I}.
\label{eq:strongMarginal}
\end{equation}
This is because for any (not necessarily separable) quantum state $\Phi_{AB}$ 
satisfying $\Tr(V_{AB}\Phi_{AB})=1$, Eq.~\eqref{eq:strongMarginal} implies the 
validity of Eq.~\eqref{eq:marginalPure3}.
Hence, this replacement will lead to an equivalent result as in 
Theorem~\ref{thm:marginalPure}.  However, when considering relaxations of the 
optimization in Eq.~\eqref{eq:marginalPure1} by replacing the separability 
constraint in Eq.~\eqref{eq:marginalPure2} with some entanglement criteria, 
Eq.~\eqref{eq:strongMarginal} may be
strictly stronger for certain marginal problems.

Second, if one finds that $\Tr(V_{AB}\Phi_{AB})=1$, this is equivalent to 
$V_{AB}\Phi_{AB}=\Phi_{AB}$, as the largest eigenvalue of $V_{AB}$ is one.  
Physically, this means that $\Phi_{AB}$ is a two-party state acting on the 
symmetric subspace only. Hence, Theorem~\ref{thm:marginalPure} is also 
equivalent to the feasibility problem
\begin{align}
    &\findover && \Phi_{AB}\in\Sep \label{eq:marginalPureFeasibility1}\\
    &\subto && V_{AB}\Phi_{AB}=\Phi_{AB},~\Tr(\Phi_{AB})=1,\label{eq:marginalPureFeasibility2}\\
    &       && \Tr_{A_\bI}(\Phi_{AB})=\rho_I\otimes\Tr_A(\Phi_{AB})\FA 
    I\in\mathcal{I}. \label{eq:marginalPureFeasibility3}
\end{align}
Furthermore, any feasible state $\Phi_{AB}$ can be used for constructing the 
global state $\ket{\varphi}$ with the desired marginals, as the proof of 
Theorem~\ref{thm:marginalPure} implies any pure state in the separable 
decomposition of $\Phi_{AB}$ can give a desired global state.

Third, the separability condition in the optimization 
Eq.~\eqref{eq:marginalPure2} is usually not easy to characterize, hence 
relaxations of the problem need to be considered.  The first candidate is the 
positive partial transpose (PPT) criterion \cite{Peres1996,Horodecki.etal1996}, 
which is an SDP relaxation of the optimization in Eq.~\eqref{eq:marginalPure1}.  
The PPT relaxation provides a pretty good approximation when the local 
dimension and the number of parties are small. In the following, inspired by 
the symmetric extension criterion \cite{Doherty.etal2002}, we propose 
a multi-party extension method and obtain a complete hierarchy for the marginal 
problem.

\vspace{1em}
\noindent\textbf{The hierarchy for the marginal problem.}
%
\begin{figure}
  \centering
  \includegraphics[width=.45\textwidth]{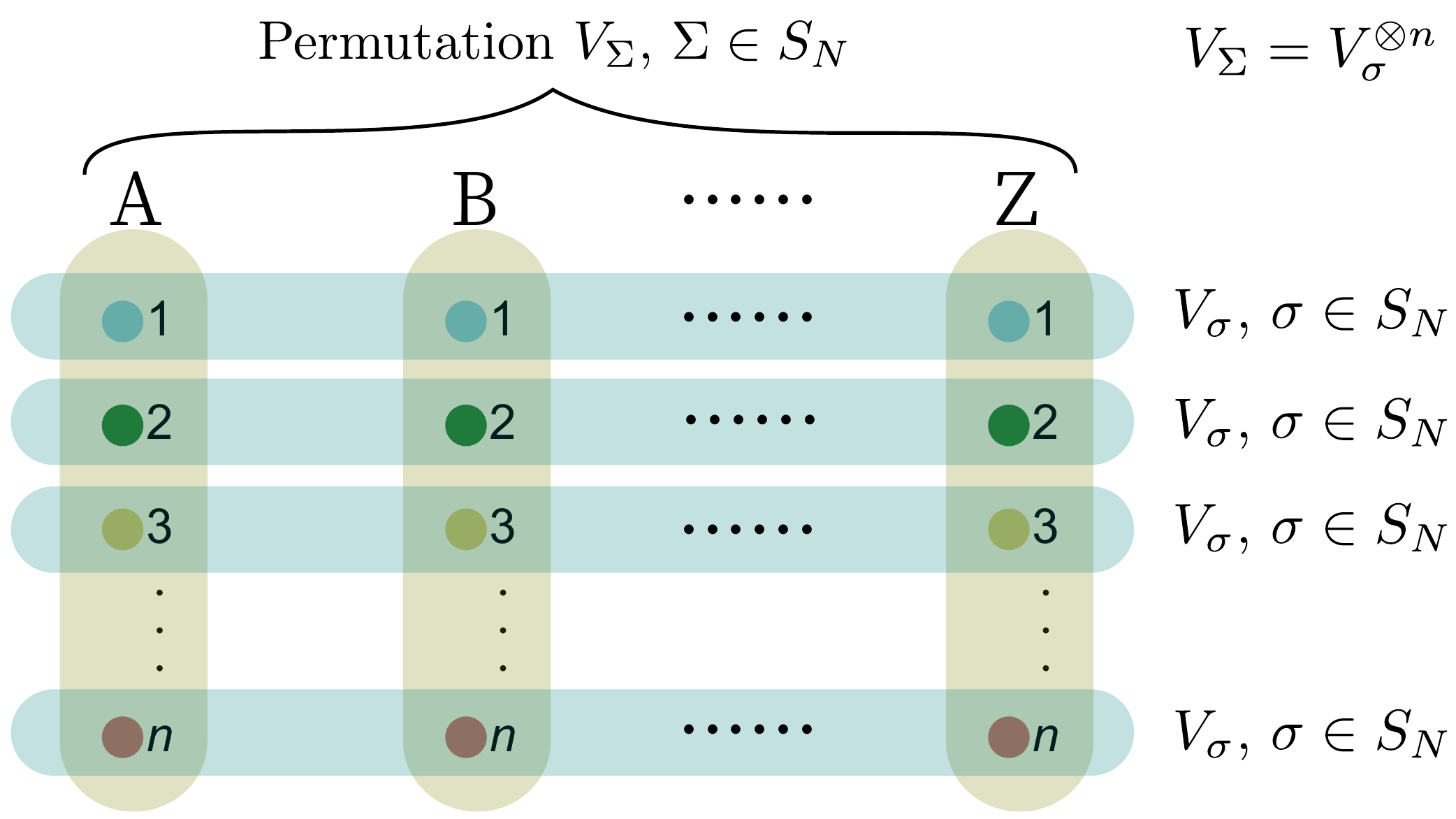}
  \caption{\textbf{Complete hierarchy for the marginal problem.} In order to 
    formulate the hierarchy for the marginal problem, one extends the two 
    copies in Fig.~\ref{fig:twocopy} to an arbitrary number of copies $N$. If 
    the marginal problem has a solution $\ket{\varphi}$, then there are 
    multi-party extensions $\Phi_{AB\cdots Z}$
    in the symmetric subspace specified by $V_\Sigma=V_\sigma^{\otimes n}$ for 
  any number of copies, obeying some semidefinite constraints.}
  \label{fig:Ncopy}
\end{figure}
%
In order to generalize Theorem~\ref{thm:marginalPure} we first need to extend 
$\cC_2$ in Eq.~(\ref{eq:cCN}) from two to an arbitrary number of copies of 
$\rho$. That is, we define $\cC_N=\conv\{\rho^{\otimes N} \mid \rho\in\cC\}$.
Second, we introduce the notion of the symmetric subspace. We denote the 
$N$ parties as $A,B,\dots,Z$, and each of them owns an $n$-body quantum system. 
For any $\mathcal{H}^{\otimes N}:=\mathcal{H}_A\otimes\mathcal{H}_B\otimes\cdots\otimes\mathcal{H}_Z$, 
the symmetric subspace is defined as
\begin{equation}
  \left\{\ket{\Psi}\in\mathcal{H}^{\otimes N}~\Big|~
  V_{\Sigma}\ket{\Psi}=\ket{\Psi} \FA \Sigma\in S_N\right\},
  \label{eq:symmetricVector}
\end{equation}
where $S_N$ is the permutation group over $N$ symbols and $V_{\Sigma}$ are the 
corresponding operators on the $N$ parties $A,B,\dots,Z$; see 
Fig.~\ref{fig:Ncopy}.  Let $P_N^+$ denote the orthogonal projector onto the 
symmetric subspace of $\mathcal{H}^{\otimes N}$.  $P_N^+$ can be explicitly 
written as
\begin{equation}
  P_N^+=\frac{1}{N!}\sum_{\Sigma\in S_N}V_\Sigma.
  \label{eq:symmetricProjector}
\end{equation}
In particular, for two parties we have the well-known relation 
$P_2^+=(\I_{AB}+V_{AB})/2$, which implies that $\Tr(V_{AB}\Phi_{AB})=1$ 
if and only if $\Tr(P_2^+\Phi_{AB})=1$. Also, $V_{AB}\Phi_{AB}=\Phi_{AB}$ 
is equivalent to $P_2^+\Phi_{AB}P_2^+=\Phi_{AB}$.
Hereafter, without ambiguity, we will use $P_N^+$ to denote both the symmetric 
subspace and the corresponding orthogonal projector.

Suppose that there exists a pure state $\rho\in\cC$. It is easy to see that 
$\Phi_{AB\cdots Z}=\rho^{\otimes N}$ satisfies
\begin{align}
    &P_N^+\Phi_{AB\cdots Z}P_N^+=\Phi_{AB\cdots Z},\label{eq:kcopynecessary1}\\
    &\Phi_{AB\cdots Z}\in\Sep,~\Tr(\Phi_{AB\cdots Z})=1,\label{eq:kcopynecessary2}\\
    &\Tr_{A_\bI}(\Phi_{AB\cdots Z})=\rho_I\otimes\Tr_A(\Phi_{AB\cdots Z})\FA 
    I\in\mathcal{I}.\label{eq:kcopynecessary3}
\end{align}
Here, the separability can be understood as either full separability or 
biseparability, since they are equivalent in the symmetric subspace 
\cite{Toth.Guehne2009}. Relaxing $\Phi_{AB\cdots Z}\in\Sep$, we obtain 
a complete hierarchy for the quantum marginal problem:

\begin{theorem}
  There exists a pure quantum state $\ket{\varphi}$ that satisfies 
  $\Tr_\bI(\ket{\varphi}\bra{\varphi})=\rho_I$ for all $I\in\mathcal{I}$ if and 
  only if for all $N\ge 2$ there exists an $N$-party quantum state 
  $\Phi_{AB\cdots Z}$ such that
    \begin{align}
      &P_N^+\Phi_{AB\cdots Z}P_N^+=\Phi_{AB\cdots Z},\label{eq:marginalPureHierarchy1}\\
      &\Phi_{AB\cdots Z}\ge 0,~\Tr(\Phi_{AB\cdots Z})=1,\label{eq:marginalPureHierarchy2}\\
      &\Tr_{A_\bI}(\Phi_{AB\cdots Z})=\rho_I\otimes\Tr_A(\Phi_{AB\cdots Z})\FA 
      I\in\mathcal{I}.\label{eq:marginalPureHierarchy3}
    \end{align}
    Each step of this hierarchy is a semidefinite feasibility problem, and the conditions
    become more restrictive if $N$ increases.
\label{thm:marginalPureHierarchy}
\end{theorem}
The proof of Theorem~\ref{thm:marginalPureHierarchy} is shown in the Methods 
section.  Notably, we can add any criterion of full separability, e.g., the PPT 
criterion for all bipartitions, as extra constraints to the feasibility 
problem. Then, Theorem~\ref{thm:marginalPureHierarchy} still provides
a complete hierarchy for the quantum 
marginal problem. In addition, the quantum marginal problems of practical 
interest are usually highly symmetric.  These symmetries can be utilized to 
largely simplify the problems in Theorems~\ref{thm:marginalPure} and 
\ref{thm:marginalPureHierarchy}. Indeed, taking advantage of symmetries is 
usually necessary for practical applications, because the general quantum marginal 
problem is QMA-complete \cite{Liu2006,Liu.etal2007}. Notably, even for non-overlapping 
marginals, despite recent progress in 
Refs.~\cite{Buergisser.etal2017,Buergisser.etal2018,Buergisser.etal2018b}, 
it is still an open problem whether there exists a polynomial-time algorithm.  
In the following, we illustrate how symmetry can drastically simplify quantum 
marginal problems with the existence problem of AME states.

\vspace{1em}
\noindent\textbf{Absolutely maximally entangled states.}
%
We first recall the definition of AME states. An $n$-qudit state 
$\ket{\psi}$ is called an AME state, denoted as $\AME(n,d)$, if 
it satisfies
\begin{equation}
  \Tr_\bI(\ket{\psi}\bra{\psi})=\frac{\I_{d^r}}{d^r} \FA I\in\mathcal{I}_r,
  \label{eq:defAME}
\end{equation}
where $\mathcal{I}_r=\{I\subset [n]\mid \abs{I}=r\}$ and $r=\lfloor 
n/2\rfloor$. Thus, Eq.~\eqref{eq:marginalPureFeasibility1} implies that an 
$\AME(n,d)$ exists if and only if the following problem is feasible,
\begin{align}
    &\findover && \Phi_{AB}\in\Sep \label{eq:optAME1}\\
    &\subto && \Tr(\Phi_{AB})=1,~V_{AB}\Phi_{AB}=\Phi_{AB}, \label{eq:optAME2}\\
    &       && \Tr_{A_\bI}(\Phi_{AB})=\frac{\I_{d^r}}{d^r}
	       \otimes\Tr_A(\Phi_{AB})\FA I\in\mathcal{I}_r. \label{eq:optAME3}
  \end{align}
Direct evaluation of the problem  is usually difficult, 
because the dimension of $\Phi_{AB}$ is $d^{2n}\times d^{2n}$, which is already 
very large for the simplest cases. For instance, for the $4$-qubit case,  
the size of $\Phi_{AB}$ is $256\times 256$.

To resolve this size issue, we investigate the symmetries that 
can be used to simplify the feasibility problem. Let $\mathcal{X}$ 
denote the set of $\Phi_{AB}$ that satisfy the constraints in 
Eqs.~(\ref{eq:optAME1},\,\ref{eq:optAME2},\,\ref{eq:optAME3}). If we find 
a unitary group $G$ such that for all $g\in G$ and $\Phi_{AB}\in\mathcal{X}$ we 
have that
\begin{equation}
  g\Phi_{AB} g^\dagger\in\mathcal{X}.
  \label{eq:symmetry}
\end{equation}
Then, the convexity of $\mathcal{X}$ implies that we can add a symmetry 
constraint to the constraints in Eqs.~(\ref{eq:optAME1},\,\ref{eq:optAME2},\,%
\ref{eq:optAME3}), namely,
\begin{equation}
   g\Phi_{AB}g^\dagger=\Phi_{AB}\FA g\in G.
  \label{eq:optAMESymm}
\end{equation}

In the following, we will show that the symmetries of the set of 
AME states (if they exist for given $n$ and $d$) are restrictive 
enough to leave only a single unique candidate for $\Phi_{AB}$, for which 
separability needs to be checked. The set of $\AME(n,d)$ is invariant 
under local unitaries and permutations on the $n$ particles, so by 
Theorem~\ref{thm:marginalPure} (or by direct verification) the following two 
classes of unitaries satisfy Eq.~\eqref{eq:symmetry},
\begin{align}
  \label{eq:group1}
  &U_1\otimes\dots\otimes U_n\otimes U_1\otimes\dots\otimes U_n \FA U_i\in 
  SU(d),\\
  \label{eq:group2}
  &\pi\otimes\pi \FA \pi\in S_n,
\end{align}
where $\pi=\pi(A_1,A_2,\dots,A_n)=\pi(B_1,B_2,\dots,B_n)$ denotes the 
permutation operators on $\mathcal{H}_A$ and $\mathcal{H}_B$.
Note that the $U_i$ in Eq.~\eqref{eq:group1} can be different.

First, let us view $V_{AB}$ and $\Phi_{AB}$ as $V_{12\dots n}$ and 
$\Phi_{12\dots n}$, where $i$ labels the subsystems $A_iB_i$.
Hereafter, without ambiguity, we will omit the subscripts of
\begin{equation}
  \I:=\I_{d^2},\quad V:=V_{A_iB_i},
  \label{eq:shortForm}
\end{equation}
for simplicity. From this perspective, $V_{AB}$ can be written as $V^{\otimes 
n}$, and the symmetries in Eqs.~(\ref{eq:group1},\,\ref{eq:group2}) can be 
written as $\bigotimes_{i=1}^n(U_i\otimes U_i)$ for $U_i\in SU(d)$ and 
$\Pi=\Pi(A_1B_1,A_2B_2,\dots,A_nB_n)$ for $\Pi\in S_n$, respectively.  
According to Werner's result~\cite{Werner1989}, a $(U\otimes U)$-invariant 
Hermitian operator must be of the form $\alpha\I+\beta V$ with 
$\alpha,\beta\in\dR$. This implies that a $\left[\bigotimes_{i=1}^n(U_i\otimes 
U_i)\right]$-invariant state must be  a linear combination of operators of the 
form
\begin{equation}
  \bigotimes_{i=1}^n(\alpha_i\I+\beta_iV)
  \FA\alpha_i,\beta_i\in\dR.
  \label{eq:linearSpan}
\end{equation}
In addition, we take advantage of the permutation symmetry under $\Pi\in S_n$ 
to write any invariant $\Phi_{AB}$ as
\begin{equation}
  \Phi_{AB}=\sum_{i=0}^nx_i\cP\{V^{\otimes i}\otimes\I^{\otimes (n-i)}\},
  \label{eq:PhiSymm}
\end{equation}
where $\cP$ represents the sum over all possible permutations that give different terms, 
e.g., $\cP\{V\otimes\I\otimes\I\}=V\otimes\I\otimes\I+
\I\otimes V\otimes\I+\I\otimes\I\otimes V$.

Inserting this ansatz in Eqs.~(\ref{eq:optAME2},\,\ref{eq:optAME3}) one can 
show by brute force calculation that the $x_i$ are uniquely determined and 
given by
\begin{equation}
  x_i=\frac{(-1)^i}{(d^2-1)^n}\sum_{l=0}^n\sum_{k=0}^l
  \frac{(-1)^l\binom{i}{k}\binom{n-i}{l-k}}
  {\min\{d^{i+2l-2k}, d^{n+i-2k}\}},
  \label{eq:xi}
\end{equation}
where we use the convention that $\binom{i}{j}=0$ when $j<0$ or $j>i$;
see \Appendix~\ref{app:uniqueness} for details. This means that the two-party 
extension under the symmetries is independent of the specific AME state, which 
is an interesting structural result considering that there exist even infinite 
families of $\AME(n,d)$ states that are not SLOCC equivalent 
\cite{Burchardt.Raissi2020}.  Together with Theorem~\ref{thm:marginalPure}, 
this result implies that an AME state exists if and only if $\Phi_{AB}$ is 
a separable quantum state.

\begin{theorem}
  An $\AME(n,d)$ state exists if and only if the operator $\Phi_{AB}$ defined 
  by Eqs.~(\ref{eq:PhiSymm},\,\ref{eq:xi}) is a separable state w.r.t. the 
  bipartition $(A|B)=(A_1A_2\dots A_n|B_1B_2\dots B_n)$.
  \label{thm:AME}
\end{theorem}

To check the separability of $\Phi_{AB}$, we first consider the positivity 
condition and the PPT condition. It is easy to see that $\Phi_{AB}$ can be 
written as
\begin{equation}
  \Phi_{AB}=\sum_{i=0}^np_i\cP\left\{P_+^{\otimes(n-i)}\otimes P_-^{\otimes 
  i}\right\},
  \label{eq:PhiSymAsy}
\end{equation}
and $\Phi_{AB}^{T_B}$ can be written as
\begin{equation}
  \Phi_{AB}^{T_B}=\sum_{i=0}^nq_i\cP\left\{P_\phi^{\otimes(n-i)}\otimes 
  P_\perp^{\otimes i}\right\},
  \label{eq:PhiTSymAsy}
\end{equation}
where
\begin{equation}
  P_\pm=\frac{1}{2}(\I\pm V),
  ~P_\phi=\ket{\phi^+}\bra{\phi^+},
  ~P_\perp=\I-P_\phi,
    \label{eq:ppmphi}
\end{equation}
with $\ket{\phi^+}=\frac{1}{\sqrt{d}}\sum_{k=1}^d\ket{k}\ket{k}$.
Here $p_i$ and $q_i$ are the eigenvalues of $\Phi_{AB}$ and $\Phi_{AB}^{T_B}$, 
respectively. Then, we can simplify the positivity condition $\Phi_{AB}\ge 0$ 
and the PPT condition $\Phi_{AB}^{T_B}\ge 0$ to
\begin{align}
  \label{eq:consPosPPT1}
  &\sum_{l=0}^n\sum_{k=0}^l\frac{(-1)^k\binom{i}{k}\binom{n-i}{l-k}}
  {\min\{d^l,d^{n-l}\}}\ge 0,\\
  \label{eq:consPosPPT2}
  &\sum_{k=0}^i\frac{(-1)^k\binom{i}{k}}
  {\min\{d^{2(n+k-i)},d^n\}}\ge 0,
\end{align}
for all $i=0,1,2,\dots,n$. Note that the latter inequality is trivial for $i\le 
r$.

The explicit form of $p_i$ and $q_i$ and the proof of the conditions in 
Eqs.~(\ref{eq:consPosPPT1},\,\ref{eq:consPosPPT2}) are shown in 
\Appendix~\ref{app:positivity}.  The positivity and PPT conditions can already 
rule out the existence of many AME states.  Actually, they can reproduce all 
the known nonexistence results \cite{Huber.Wyderka2018} except $\AME(7,2)$ 
\cite{Huber.etal2017}.  To get a higher-order approximation, we provide 
a general framework for performing the symmetric extension in 
\Appendices~\ref{app:SEGen} and \ref{app:SEAME}.

As the open problem of the existence of $\AME(4,6)$ is of particular interest 
in the quantum 
information community \cite{Horodecki.etal2020,OQP35}, we explicitly express it 
as the following corollary.
\begin{corollary}
  An $\AME(4,6)$ state exists if and only if the quantum state
  \begin{equation}
    \Phi_{AB}= \frac{1}{2\cdot 6^4} \Big(
    \frac{P_+^{\otimes 4}}{343}
    +\frac{\cP\big\{P_+^{\otimes 2}
    \otimes P_-^{\otimes 2}\big\}}{315}
    +\frac{P_-^{\otimes 4}}{375} \Big),
    \label{eq:Phi46}
  \end{equation}
  is separable, or equivalently,
  \begin{equation}
    \Phi_{AB}^{T_B}= \frac{1}{6^4} \Big(
    P_\phi^{\otimes 4}
    +\frac{\cP\big\{P_\phi\otimes
    P_\perp^{\otimes 3}\big\}}{35^2}
    +\frac{33 P_\perp^{\otimes 4}}{35^3} \Big),
    \label{eq:Phi46T}
  \end{equation}
  is separable w.r.t.  bipartition $(A|B)$.
  \label{cor:AME46}
\end{corollary}
At the moment, we are unable to decide separability of these states; in
\Appendix~\ref{app:fail} we provide a short discussion of this problem.

\vspace{1em}
\noindent\textbf{Quantum codes.}
%
As another application, we show that our method can also be used to analyze the 
existence of quantum error correcting codes. For simplicity, we only consider 
pure quantum codes \cite{Rains1999} in the text; see Methods for the general 
case.  Our starting point is the fact that pure quantum codes are closely 
related to $m$-uniform states \cite{Scott2004}.  More precisely, an 
$((n,K,m+1))_d$ pure code exists if and only if there exists a $K$-dimensional 
subspace $\mathcal{Q}$ of 
$\mathcal{H}=\bigotimes_{i=1}^n\mathcal{H}_i=(\dC^d)^{\otimes n}$ such that all 
states in $\mathcal{Q}$ are $m$-uniform, i.e., for all 
$\ket{\varphi}\in\mathcal{Q}$
\begin{equation}
  \Tr_{\bI}(\ket{\varphi}\bra{\varphi})=\frac{\I_{d^m}}{d^m}
  \FA I\in\mathcal{I}_m,
  \label{eq:coding}
\end{equation}
where $\mathcal{I}_m=\{I\in [n]\mid \abs{I}=m\}$ and $\bI=[n]\setminus I$.  
The existence of $((n,1,m+1))_d$ pure codes reduces to the existence of 
$m$-uniform states, for which the methods from the last section are directly 
applicable. Here, we show that the existence of $((n,K,m+1))_d$ pure codes can 
still be written as a marginal problem if $K>1$. To do so, we define an 
auxiliary system $\mathcal{H}_0=\dC^K$ and let 
$\widetilde{\mathcal{H}}=\mathcal{H}_0\otimes\mathcal{H}
=\bigotimes_{i=0}^n\mathcal{H}_i=\dC^K\otimes(\dC^d)^{\otimes n}$. Now, we can 
write the existence of $((n,K,m+1))_d$ pure codes as a marginal problem on 
$\widetilde{\mathcal{H}}$.
\begin{lemma}
  A quantum $((n,K,m+1))_d$ pure code exists if and only if there exists 
  a quantum state $\ket{Q}$ in $\widetilde{\mathcal{H}}$ such that
  \begin{equation}
    \Tr_\bI(\ket{Q}\bra{Q})=\frac{\I_{Kd^m}}{Kd^m}
    \FA I\in\mathcal{I}_m,
    \label{eq:Qcode}
  \end{equation}
  where $\bI$ is still defined as $\{1,2,\dots,n\}\setminus I$.
  \label{lem:Qcode}
\end{lemma}

\begin{proof}
  We first show the necessity part. Suppose that a $((n,K,m+1))_d$ code with corresponding 
  subspace $\mathcal{Q}$ exists. We define an entangled state $\ket{Q}$ in 
  $\mathcal{H}_0\otimes\mathcal{Q}\subset\widetilde{\mathcal{H}}$ as
  \begin{equation}
    \ket{Q}=\frac{1}{\sqrt{K}}\sum_{k=1}^K\ket{k}\ket{k_L},
    \label{eq:defQ}
  \end{equation}
  where $\{\ket{k}\}_{k=1}^K$ and $\{\ket{k_L}\}_{k=1}^K$ are orthonormal bases 
  for $\mathcal{H}_0$ and $\mathcal{Q}$, respectively. Then for any pure state 
  $\ket{a}$ in $\mathcal{H}_0$, $\sqrt{K}\braket{a}{Q}\in\mathcal{Q}$. Hence, 
  Eq.~\eqref{eq:coding} implies that
  \begin{equation}
    \Tr_0[\Tr_\bI(\ket{a}\bra{a}\otimes\I_{d^n}\ket{Q}\bra{Q})]
    =\frac{\I_{d^m}}{Kd^m}
    \FA I\in\mathcal{I}_m,
    \label{eq:atrace}
  \end{equation}
  for all $\ket{a}$ in $\mathcal{H}_0$, which in turn implies Eq.~\eqref{eq:Qcode}.

  To prove the sufficiency part, let $\mathcal{Q}$ be the space generated by 
  the pure states $\ket{\varphi_a}=\sqrt{K}\braket{a}{Q}$ for all $\ket{a}$ in 
  $\mathcal{H}_0$.  Then, Eq.~\eqref{eq:Qcode} implies that all 
  $\ket{\varphi_a}$ are $m$-uniform states. Furthermore, from 
  $\rank(\Tr_0(\ket{Q}\bra{Q}))=\rank(\Tr_{12\cdots n} 
  (\ket{Q}\bra{Q}))=\rank(\I_K/K)=K$ it follows that 
  $\mathcal{Q}$ is a $K$-dimensional subspace.
\end{proof}

Thus, Theorem~\ref{thm:marginalPure} gives a necessary and sufficient 
condition for the existence of $((n,K,m+1))_d$ pure codes.

\begin{proposition}
  A quantum $((n,K,m+1))_d$ pure code exists if and only if there exists 
  $\Phi_{AB}$ in $\widetilde{\mathcal{H}}_A\otimes 
  \widetilde{\mathcal{H}}_B=[\dC^K\otimes(\dC^d)^{\otimes n}]^{\otimes 2}$ such 
  that
  \begin{align}
    \label{eq:code1}
    &\Phi_{AB}\in\Sep,~V_{AB}
    \Phi_{AB}=\Phi_{AB},
    ~\Tr(\Phi_{AB})=1,\\
    \label{eq:code2}
    &\Tr_{A_\bI}(\Phi_{AB})
    =\frac{\I_{Kd^m}}{Kd^m}\otimes\Tr_A(\Phi_{AB})
    \FA I\in\mathcal{I}_m,
  \end{align}
  where $\Sep$ denotes the set of separable states w.r.t.~the bipartition 
  $(A|B)=(A_0A_1\cdots A_n|B_0B_1\cdots B_n)$, $V_{AB}$ is the swap operator 
  between $\widetilde{\mathcal{H}}_A$ and $\widetilde{\mathcal{H}}_B$, and 
  $A_\bI$ denotes all subsystems $A_i$ for $i\in\bI$.
  \label{thm:code}
\end{proposition}

Furthermore, the multi-party extension and symmetrization techniques that we 
developed for AME states can be easily adapted to the quantum error correcting 
codes. For instance, the PPT relaxation can be written as a linear program and 
the symmetric extensions can be written as SDPs. An important difference is 
that the symmetrized $\Phi_{AB}$ for quantum error correcting codes is no 
longer uniquely determined by the marginals in general. Finally, we would like 
to mention that Lemma~\ref{lem:Qcode} is of independent interest on its own.  
For example, Eq.~\eqref{eq:Qcode} implies that $Kd^m\le\sqrt{Kd^n}$, as 
$\rank(\Tr_\bI(\ket{Q}\bra{Q}))\le\sqrt{\dim(\widetilde{\mathcal{H}})}$. This 
provides a simple proof for the quantum Singleton bound 
\cite{Knill.Laflamme1997,Rains1999} $K\le d^{n-2m}$ for pure codes.

\section{Discussion}
%
We have shown that the marginal problem for multiparticle quantum systems
is closely related to the problem of entanglement and separability for 
two-party systems. More precisely, we have shown that the existence of
a pure multiparticle state with given marginals can be reformulated as the
existence of a two-party separable state with additional semidefinite 
constraints.  This allows for further refinements: First, one may use the 
multi-party extension technique to develop a complete hierarchy for the quantum 
marginal problem.  Second, one can use symmetries of the original marginal 
problem, to restrict the search of the two-party separable state further. For 
the AME problem, this allows us to determine a unique candidate for the state, and 
it remains to check its separability properties. Finally, the approach can be 
extended to characterize the existence of quantum codes.

Our work provides new insights in several subfields of quantum information 
theory. First, it may provide a significant step towards solving the problem
of the existence of the $\AME(4,6)$ state or quantum orthogonal Latin 
squares, a problem which has been highlighted as an outstanding problem in 
quantum information theory \cite{Horodecki.etal2020}. Second, there are already 
a variety of results on the separability problem, and in the future, these can be
used to study marginal problems in various situations. Finally, it would be
interesting to extend our work to other versions of the marginal problem, 
e.g., in fermionic systems or with a relaxed version of the purity constraint. 
We believe that our approach can also lead to progress in these cases.

\section{Methods}
\noindent\textbf{Proof of Theorem~\ref{thm:marginalPureHierarchy}.}
%
To prove Theorem~\ref{thm:marginalPureHierarchy}, we take advantage the 
following lemma, which can be viewed as a special case of the quantum de 
Finetti theorem \cite{Christandl.etal2007}.

\begin{lemma}
  Let $\rho_N$ be an $N$-party quantum state in the symmetric subspace $P_N^+$, 
  then there exists a $k$-party quantum state
  \begin{equation}
    \sigma_k=\sum_\mu p_\mu\ket{\psi_\mu}\bra{\psi_\mu}^{\otimes k},
    \label{eq:ksep}
  \end{equation}
  i.e., a fully separable state in $P_k^+$, such that
  \begin{equation}
    \norm{\Tr_{N-k}(\rho_N)-\sigma_k}
    \le\frac{4kD}{N},
    \label{eq:deFinetti}
  \end{equation}
  where $\norm{\cdot}$ is the trace norm and $D$ is the local dimension.
  \label{lem:deFinetti}
\end{lemma}

The necessity part of Theorem~\ref{thm:marginalPureHierarchy} is obvious.  
Hence, we only need to prove the sufficiency part, i.e., that the existence of 
an $N$-party quantum state $\Phi_{AB\cdots Z}$ for arbitrary $N$ implies the 
existence of $\ket{\varphi}$.
Let $\Phi^N_{AB}=\Tr_{C\cdots Z}(\Phi_{ABC\cdots Z})$, then $\Phi^N_{AB}$ 
satisfies
\begin{equation}
  \Tr(\Phi^N_{AB})=1,\quad
  \Tr_{A_\bI}(\Phi^N_{AB})=\rho_I\otimes\Tr_A(\Phi^N_{AB})
  \FA I\in\mathcal{I}.
  \label{eq:reducedPhi}
\end{equation}
Further, Lemma~\ref{lem:deFinetti} implies that there exist separable states
$\widetilde{\Phi}^N_{AB}$ such that
\begin{align}
  &V_{AB}\widetilde{\Phi}^N_{AB}=\widetilde{\Phi}^N_{AB},
  \label{eq:PhideFinetti1}\\
  &\norm{\Phi^N_{AB}-\widetilde{\Phi}^N_{AB}}\le\frac{8D}{N}.
  \label{eq:PhideFinetti2}
\end{align}
As the set of quantum states for any fixed dimension is compact, we can choose 
a convergent subsequence $\Phi^{N_i}_{AB}$ of the sequence $\Phi^{N}_{AB}$.  
Thus, Eq.~\eqref{eq:PhideFinetti2} implies that
\begin{equation}
  \Phi_{AB}:=\lim_{i\to+\infty}\Phi^{N_i}_{AB}
  =\lim_{i\to+\infty}\widetilde{\Phi}^{N_i}_{AB}.
  \label{eq:convergentPhi}
\end{equation}
Thus, Eqs.~(\ref{eq:reducedPhi},\,\ref{eq:PhideFinetti1}) and the fact that the 
set of separable states is closed imply that $\Phi_{AB}$ satisfies all 
constraints in Eqs.~(\ref{eq:marginalPureFeasibility1},\,%
\ref{eq:marginalPureFeasibility2},\,\ref{eq:marginalPureFeasibility3}).
Then, Theorem~\ref{thm:marginalPureHierarchy} follows directly from 
Theorem~\ref{thm:marginalPure}.

\vspace{1em}
\noindent\textbf{General quantum codes.}
%
In general, a quantum $((n,K,m+1))_d$ code exists if and only if there exists 
a $K$-dimensional subspace $\mathcal{Q}$ of 
$\mathcal{H}=\bigotimes_{i=1}^n\mathcal{H}_i=(\dC^d)^{\otimes n}$ such that
for all $\ket{\varphi}\in\mathcal{Q}$
\begin{equation}
  \Tr_{\bI}(\ket{\varphi}\bra{\varphi})=\rho_I
  \FA I\in\mathcal{I}_m,
  \label{eq:codingImpure}
\end{equation}
where $\rho_I$ are marginals that are arbitrary but independent of 
$\ket{\varphi}$, $\mathcal{I}_m=\{I\in [n]\mid \abs{I}=m\}$, and 
$\bI=[n]\setminus I=\{1,2,\dots,n\}\setminus I$.  Similar to the case of pure 
codes, we can prove the following lemma.
\begin{lemma}
  A quantum $((n,K,m+1))_d$ code exists if and only if there exists a quantum 
  state $\ket{Q}$ in $\widetilde{\mathcal{H}}$ and marginal states $\rho_I$ 
  such that
  \begin{equation}
    \Tr_\bI(\ket{Q}\bra{Q})=\frac{\I_K}{K}\otimes\rho_I
    \FA I\in\mathcal{I}_m,
    \label{eq:QcodeImpure}
  \end{equation}
  where $\widetilde{\mathcal{H}}=\mathcal{H}_0\otimes\mathcal{H}
  =\bigotimes_{i=0}^n\mathcal{H}_i=\dC^K\otimes(\dC^d)^{\otimes n}$ and $\bI$ 
  is defined as $[n]\setminus I=\{1,2,\dots,n\}\setminus I$.
  \label{lem:QcodeImpure}
\end{lemma}
If the marginals $\rho_I$ are given like in the case of pure codes, the problem 
reduces to a marginal problem. However, to ensure the existence of 
$((n,K,m+1))_d$ codes, an arbitrary set of marginals is sufficient. This makes 
the problem no longer a marginal problem, however, we can circumvent this issue 
by observing that Eq.~\eqref{eq:QcodeImpure} is equivalent to
\begin{equation}
  \Tr_0[(M_0\otimes\I_I)\Tr_\bI(\ket{Q}\bra{Q})]=0
  \FA I\in\mathcal{I}_m,
  \label{eq:QcodeImpureLinear}
\end{equation}
for all $M_0$ such that $\Tr(M_0)=0$. Moreover, we can choose an arbitrary 
basis $\mathcal{B}$ for $\{M_0\mid\Tr(M_0)=0,~M_0^\dagger=M_0\}$. Then, with 
the general result on rank-constrained optimization from 
Ref.~\cite{Yu.etal2020b}, we obtain the following theorem, and similar to the 
AME existence problem, a complete hierarchy can be constructed using the 
symmetric extension technique.
\begin{proposition}
  A quantum $((n,K,m+1))_d$ code exists if and only if there exists $\Phi_{AB}$ 
  in $\widetilde{\mathcal{H}}_A\otimes 
  \widetilde{\mathcal{H}}_B=[\dC^K\otimes(\dC^d)^{\otimes n}]^{\otimes 2}$ such 
  that
  \begin{align}
    \label{eq:codeImpure1}
    &\Phi_{AB}\in\Sep,\,V_{AB}\Phi_{AB}=\Phi_{AB},
    \,\Tr(\Phi_{AB})=1,\\
    \label{eq:codeImpure2}
    &\Tr_{A_0}\Tr_{A_\bI}[(M_{A_0}\otimes\I_{A_0^c})\Phi_{AB}]=0,
  \end{align}
  for all $I\in\mathcal{I}_m$ and $M_{A_0}\in\mathcal{B}$,
  where the $\Sep$ means the separability with respect to the bipartition 
  $(A|B)=(A_0A_1\cdots A_n|B_0B_1\cdots B_n)$, $V_{AB}$ is the swap operator 
  between $\widetilde{\mathcal{H}}_A$ and $\widetilde{\mathcal{H}}_B$, $A_\bI$ 
  denotes all subsystems $A_i$ for $i\in\bI$, and $\I_{A_0^c}$ denote the 
  identity operator on $AB\setminus A_0=A_1A_2\cdots A_nB_0B_1B_2\cdots B_n$.
  \label{thm:codeImpure}
\end{proposition}
By noticing that the set of $((n,K,m+1))_d$ (pure or general) codes, or rather, 
the set of states $\ket{Q}$, is invariant under local unitaries and 
permutations on the bodies $123\cdots n$, we can assume that $\Phi_{AB}$ is 
invariant under the following two classes of unitaries
\begin{align}
  \label{eq:groupCode1}
  &U_0\otimes U_1\otimes\dots\otimes U_n\otimes U_0\otimes 
  U_1\otimes\dots\otimes U_n,\\
  \label{eq:groupCode2}
  &\id_0\otimes\pi\otimes\id_0\otimes\pi.
\end{align}
for all $U_0\in SU(K)$, $U_i\in SU(d)$, and $\pi\in S_n$. Thus, the symmetrized 
$\Phi_{AB}$ is of the form
\begin{equation}
  \begin{aligned}
    \Phi_{AB}=&\I_{K^2}\otimes\sum_{i=0}^nx_i\cP\{V^{\otimes 
    i}\otimes\I^{\otimes (n-i)}\}\\
    &+V_{A_0B_0}\otimes\sum_{i=0}^ny_i\cP\{V^{\otimes i}\otimes\I^{\otimes 
    (n-i)}\},
  \end{aligned}
  \label{eq:PhiSymmCode}
\end{equation}
for $x_i,y_i\in\dR$. Hence, all the techniques we developed for AME states can be 
easily adapted to the quantum error correcting codes. For example, the PPT 
relaxation can be written as a linear program and the symmetric extension can 
be written as SDPs.



\begin{acknowledgments}
  We would like to thank Felix Huber and G\'eza T\'oth for discussions.  This 
  work was supported by the Deutsche Forschungsgemeinschaft (DFG, German 
  Research Foundation - 447948357), the ERC (Consolidator Grant 683107/TempoQ), 
  and the House of Young Talents Siegen.  N.W. acknowledges support by the 
  QuantERA grant QuICHE and the German ministry of education and research (BMBF 
  grant no. 16KIS1119K).
\end{acknowledgments}



\onecolumngrid

\appendix
\newcommand{\bu}[1]{\noindent$\bullet$~{\bf \boldmath #1}}

\section{Existence and uniqueness of the symmetrized $\Phi_{AB}$ for AME 
states}
\label{app:uniqueness}

Before proving the existence and uniqueness of the symmetrized $\Phi_{AB}$, we 
show how to simplify the constraints in 
\MEqs~(\ref{eq:optAME2},\,\ref{eq:optAME3}) by taking advantage of 
\MEq~\eqref{eq:PhiSymm}. The meaning of this simplification is two-fold: first, 
it gives an intuition about why the symmetrized $\Phi_{AB}$ is uniquely 
determined; second, it can be directly generalized to other marginal problems, 
such as the $m$-uniform states and quantum codes, in which the symmetrized 
$\Phi_{AB}$ are no longer uniquely determined. Recall the symmetrized 
$\Phi_{AB}$ is of the form
\begin{equation}
  \Phi_{AB}=\sum_{i=0}^nx_i\cP\{V^{\otimes i}\otimes\I^{\otimes (n-i)}\},
  \label{eq:PhiSymmA}
\end{equation}
then the constraints in \MEqs~(\ref{eq:optAME2},\,\ref{eq:optAME3}) can be 
simplified as follows:\\[1ex]
\bu{Normalization constraint $\Tr(\Phi_{AB})=1$:}
\begin{equation}
     \Tr(\Phi_{AB})
     =\Tr\left[\sum_{i=0}^nx_i\cP\{V^{\otimes i}\otimes
     \I^{\otimes (n-i)}\}\right]
     =\sum_{i=0}^n\binom{n}{i}d^{2n-i}x_i=1.
  \label{eq:consNormalizationSim}
\end{equation}

\bu{Symmetric subspace constraint $V_{AB}\Phi_{AB}=\Phi_{AB}$:}
\begin{equation}
  V_{AB}\Phi_{AB}=V^{\otimes n}\Phi_{AB}
  =\sum_{i=0}^nx_i\cP\{V^{\otimes (n-i)}\otimes\I^{\otimes i}\}
  =\sum_{i=0}^nx_i\cP\{V^{\otimes i}\otimes\I^{\otimes (n-i)}\},
  \label{eq:consSymmetric}
\end{equation}
which implies that
\begin{equation}
  x_i=x_{n-i} \FA i=0,1,\dots,n-r-1,
  \label{eq:consSymmetricSim}
\end{equation}
where $r=\lfloor n/2\rfloor$.

\bu{Marginal constraints 
$\Tr_{A_\bI}(\Phi_{AB})=\frac{\I_{d^r}}{d^r}\otimes\Tr_A(\Phi_{AB})$:}\\[1ex]
Because $\Phi_{AB}$ is invariant under permutations $\Pi\in S_n$, it is 
sufficient to consider $\bI=\{1,2,\dots,n-r\}$. Further, as 
$\frac{\I_{d^r}}{d^r}\otimes\Tr_A(\Phi_{AB}) \propto\I_{d^{n+r}}$, it must also 
hold
that $\Tr_{A_\bI}(\Phi_{AB}) \propto\I_{d^{n+r}}$. Hence, all terms that 
contain $V$ in $\Tr_{A_\bI}(\Phi_{AB})$ must be zero. Thus, the marginal 
constraints 
$\Tr_{A_\bI}(\Phi_{AB})=\frac{\I_{d^r}}{d^r}\otimes\Tr_A(\Phi_{AB})$ are 
equivalent to
\begin{equation}
  \sum_{i=0}^{n-r}\binom{n-r}{i}d^{n-r-i}x_{s+i}=0 \FA s=1,2,\dots,r.
  \label{eq:consPartialTraceSim}
\end{equation}
\AEquations~(\ref{eq:consNormalizationSim},\,\ref{eq:consSymmetricSim},\,%
\ref{eq:consPartialTraceSim}) provide $n+1$ linear equations, which can 
uniquely determine the $n+1$ parameters $(x_0,x_1,\dots,x_n)$ in $\Phi_{AB}$.

To rigorously prove the existence and uniqueness of $\Phi_{AB}$ constrained by 
\AEqs~(\ref{eq:consNormalizationSim},\,\ref{eq:consSymmetricSim},\,%
\ref{eq:consPartialTraceSim}), we take advantage of the following lemma; for 
more details about the dual basis, see e.g., Ref.~\cite{Lebedev.etal2010}.
\begin{lemma}
  Let $\{\ket{x_i}\}_i$ be a basis for a finite-dimensional Hilbert space, 
  which is not required to be orthogonal or normalized. Then, there exists 
  a unique vector $\ket{y}$ satisfying the linear equations 
  $\{\braket{x_i}{y}=y_i\}_i$ for any $\{y_i\}_i$. Concretely, let 
  $\{\ket{\tilde{x}_i}\}_i$ be the dual basis for $\{\ket{x_i}\}_i$, i.e., 
  $\braket{x_i}{\tilde{x}_j}=\delta_{ij}$, then 
  $\ket{y}=\sum_iy_i\ket{\tilde{x}_i}$.
  \label{lem:basis}
\end{lemma}

First, we define $\mathcal{S}$ to be the space generated by the linearly 
independent operators
\begin{equation}
  X_i=\cP\{V^{\otimes i}\otimes\I^{\otimes (n-i)}\}
  \FA i=0,1,\dots,n,
  \label{eq:Ai}
\end{equation}
and the inner product to be the Hilbert-Schmidt inner product, e.g.,
\begin{equation}
  \mean{X_i,X_j}=\Tr(X_i^\dagger X_j)=\Tr(X_iX_j).
  \label{eq:inner}
\end{equation}
Then, $\Phi_{AB}\in\mathcal{S}$ by \AEq~\eqref{eq:PhiSymmA}.

Second, we show that if $\Phi_{AB}$ exists, then it is unique. By slightly 
modifying the derivation of \AEq~\eqref{eq:consPartialTraceSim}, it is easy to see
that the normalization constraint and the marginal constraints for $\AME(n,d)$ 
are equivalent to
\begin{equation}
  \Tr_{A_\bI B_\bI}(\Phi_{AB})=\frac{\I_{d^r}}{d^r}\otimes\frac{\I_{d^r}}{d^r}
  \FA I\in\mathcal{I}_r,
  \label{eq:condDT}
\end{equation}
which implies that
\begin{equation}
  \Tr(X_i\Phi_{AB})=\binom{n}{i}\Tr\left[V^{\otimes i}
  \frac{\I_{d^i}}{d^i}\otimes\frac{\I_{d^i}}{d^i}\right]
  =\frac{\binom{n}{i}}{d^i}
  \FA i=0,1,\dots,r.
  \label{eq:marginals}
\end{equation}
The symmetric subspace constraint $V_{AB}\Phi_{AB}=V^{\otimes 
n}\Phi_{AB}=\Phi_{AB}$ and the relation $X_iV_{AB}=X_iV^{\otimes n}=X_{n-i}$ 
imply that
\begin{equation}
  \Tr(X_i\Phi_{AB})=\Tr(X_iV_{AB}\Phi_{AB})=\Tr(X_{n-i}\Phi_{AB})
  \FA i=0,1,\dots,n.
  \label{eq:aian}
\end{equation}
Thus, we get
\begin{equation}
  \mean{X_i,\Phi_{AB}}=\Tr(X_i\Phi_{AB})
  =\frac{\binom{n}{i}}{\min\{d^i,d^{n-i}\}}
  \FA i=0,1,\dots,n.
  \label{eq:ai}
\end{equation}
which implies the uniqueness by Lemma~\ref{lem:basis}.
Furthermore, in this case, we can easily write down the dual basis 
$\{\widetilde{X}_i\}_{i=0}^n$ for $\{X_i\}_{i=0}^n$,
\begin{equation}
  \widetilde{X}_i=\frac{1}{\binom{n}{i}(d^2-1)^n}
  \cP\left\{(\I-\frac{1}{d}V)^{\otimes i}\otimes(V-\frac{1}{d}\I)^{\otimes 
  (n-i)}\right\}
  \FA i=0,1,\dots,n.
  \label{eq:AiDual}
\end{equation}
It is straightforward to check that $\Tr(\widetilde{X}_iX_j)=\delta_{ij}$. Thus, 
we can get an explicit form of $\Phi_{AB}$ from
$x_i=\Tr(\widetilde{X}_i\Phi_{AB})$,
\begin{equation}
  \begin{aligned}
    x_i=&\frac{1}{(d^2-1)^n}\Tr\left[(\I-\frac{1}{d}V)^{\otimes i}
    \otimes(V-\frac{1}{d}\I)^{\otimes (n-i)}\Phi_{AB}\right]\\
    =&\frac{1}{(d^2-1)^n}\sum_{l=0}^n\sum_{k=0}^l
    \frac{(-1)^{i+l}}{d^{i+l-2k}}\binom{i}{k}\binom{n-i}{l-k}
    \Tr\left[V^{\otimes (n-l)}\otimes\I^{\otimes l}\Phi_{AB}\right]\\
    =&\frac{(-1)^i}{(d^2-1)^n}
    \sum_{l=0}^n\sum_{k=0}^l\frac{(-1)^l\binom{i}{k}\binom{n-i}{l-k}}
    {\min\{d^{i+2l-2k}, d^{n+i-2k}\}},
  \end{aligned}
  \label{eq:xiApp}
\end{equation}
where we have used the relation
\begin{equation}
  \Tr\left[V^{\otimes (n-l)}\otimes\I^{\otimes l}
  \Phi_{AB}\right]=\frac{1}{\min\{d^l,d^{n-l}\}},
\end{equation}
whose proof is similar to \AEq~\eqref{eq:ai}.

Finally, we show the existence of $\Phi_{AB}$, i.e.,
$\Phi_{AB}$ determined by \AEq~\eqref{eq:ai} is compatible with the constraints 
in \AEqs~(\ref{eq:consNormalizationSim},\,\ref{eq:consSymmetricSim},\,%
\ref{eq:consPartialTraceSim}). To this end, we show that \AEq~\eqref{eq:ai} 
implies that $V_{AB}\Phi_{AB}=\Phi_{AB}$ and \AEq~\eqref{eq:condDT}. As 
$\Tr(X_i\Phi_{AB})=\Tr(X_{n-i} \Phi_{AB})$ by \AEq~\eqref{eq:ai} and 
$X_iV_{AB}=X_iV^{\otimes n}=X_{n-i}$, it holds that
\begin{equation}
  \Tr(X_i\Phi_{AB})=\Tr(X_iV_{AB}\Phi_{AB})
  \FA i=0,1,\dots,n.
  \label{eq:aisymm}
\end{equation}
From the uniqueness statement in Lemma~\ref{lem:basis}, it follows that 
$V_{AB}\Phi_{AB}=\Phi_{AB}$. To prove \AEq~\eqref{eq:condDT}, we define 
$\mathcal{R}$ to be the space generated by the linearly independent operators
\begin{equation}
  R_i=\cP\{V^{\otimes i}\otimes\I^{\otimes (r-i)}\}
  \FA i=0,1,\dots,r.
  \label{eq:Bi}
\end{equation}
\AEquation~\eqref{eq:ai} and the permutation symmetry of 
$\Phi_{AB}\in\mathcal{S}$ imply that
\begin{equation}
  \Tr\left[V^{\otimes i}\otimes\I^{\otimes(n-i)}\Phi_{AB}\right]=\frac{1}{d^i}
  \FA i=0,1,\dots,r.
  \label{eq:permus}
\end{equation}
Thus,
\begin{equation}
  \Tr[R_i\Tr_{A_\bI B_\bI}(\Phi_{AB})]=\binom{r}{i}\Tr\left[V^{\otimes 
  i}\otimes\I^{\otimes(n-i)}\Phi_{AB}\right]
  =\frac{\binom{r}{i}}{d^i},
  \FA i=0,1,\dots,r
  \FA I\in\mathcal{I}_r,
  \label{eq:condDTB}
\end{equation}
Furthermore, one can easily check that
\begin{equation}
  \Tr\left[R_i\frac{\I_{d^r}}{d^r}\otimes\frac{\I_{d^r}}{d^r}\right]
  =\frac{\binom{r}{i}}{d^i}
  \FA i=0,1,\dots,r.
  \label{eq:condDTId}
\end{equation}
Then, applying the uniqueness statement in Lemma~\ref{lem:basis} to 
$\mathcal{R}$ implies \AEq~\eqref{eq:condDT}. Hence, we proved the compatibility 
of $\Phi_{AB}$ with 
\AEqs~(\ref{eq:consNormalizationSim},\,\ref{eq:consSymmetricSim},\,%
\ref{eq:consPartialTraceSim}).

\section{Positivity and PPT conditions for AME state}
\label{app:positivity}

To get a closed form of the positivity and PPT conditions for AME states, we 
will use the following relations
\begin{gather}
  \Tr\left(V^{\otimes l}\otimes\I^{\otimes 
  (n-l)}\Phi_{AB}\right)=\frac{1}{\min\{d^l,d^{n-l}\}},\\
  \Tr\left(\ket{\phi^+}\bra{\phi^+}^{\otimes l}\otimes\I^{\otimes 
  (n-l)}\Phi_{AB}^{T_B}\right)=\frac{1}{\min\{d^{2l},d^n\}},
\end{gather}
where the proof of the first relation is similar to 
\AEqs~(\ref{eq:ai},\,\ref{eq:permus}) and the second relation follows from the 
observation that $\Tr(W\Phi_{AB}^{T_B})=\Tr(W^{T_B}\Phi_{AB})$. From 
\MEq~\eqref{eq:PhiSymAsy} it follows that
the positivity condition 
is equivalent to $\Tr(P_+^{\otimes (n-i)}\otimes P_-^{\otimes 
i}\Phi_{AB})\ge0$. This gives
\begin{equation}
  \begin{aligned}
    &\Tr\left[(\I+V)^{\otimes (n-i)}\otimes(\I-V)^{\otimes i}\Phi_{AB}\right]\\
    =&\Tr\left[\sum_{l=0}^n\sum_{k=0}^l(-1)^k\binom{i}{k}\binom{n-i}{l-k}
    V^{\otimes l}\otimes\I^{\otimes(n-l)}\Phi_{AB}\right]\\
    =&\sum_{l=0}^n\sum_{k=0}^l\frac{(-1)^k\binom{i}{k}\binom{n-i}{l-k}}
    {\min\{d^l,d^{n-l}\}}\ge 0 \FA i=0,1,\dots,n.
  \end{aligned}
  \label{eq:positivityFinal}
\end{equation}
Similarly due to \MEq~\eqref{eq:PhiTSymAsy}, the PPT condition is equivalent to
\begin{equation}
  \begin{aligned}
    &\Tr\left[\ket{\phi^+}\bra{\phi^+}^{\otimes(n-i)}\otimes
      (\I-\ket{\phi^+}\bra{\phi^+})^{\otimes i}\Phi_{AB}^{T_B}\right]\\
    =&\Tr\left[\sum_{k=0}^i(-1)^k\binom{i}{k}\ket{\phi^+}
      \bra{\phi^+}^{\otimes(n+k-i)}
    \otimes\I^{\otimes(i-k)}\Phi_{AB}^{T_B}\right]\\
    =&\sum_{k=0}^i\frac{(-1)^k\binom{i}{k}}
    {\min\{d^{2(n+k-i)},d^n\}}\ge 0
    \FA i=0,1,\dots,n.
  \end{aligned}
  \label{eq:PPTFinal}
\end{equation}
By noticing that
\begin{align}
  \label{eq:trace1}
  &\Tr[(\I+V)^{\otimes (n-i)}\otimes(\I-V)^{\otimes i}]
  =d^n(d+1)^{n-i}(d-1)^i\\
  \label{eq:trace2}
  &\Tr[\ket{\phi^+}\bra{\phi^+}^{\otimes(n-i)}\otimes
  (\I-\ket{\phi^+}\bra{\phi^+})^{\otimes i}]
  =(d^2-1)^i,
\end{align}
we obtain an explicit expressions for $p_i$ and $q_i$
\begin{align}
  \label{eq:qp1}
  p_i&=\frac{1}{d^n(d+1)^{n-i}(d-1)^i}
  \sum_{l=0}^n\sum_{k=0}^l\frac{(-1)^k\binom{i}{k}\binom{n-i}{l-k}}
  {\min\{d^l,d^{n-l}\}},\\
  \label{eq:qp2}
  q_i&=\frac{1}{(d^2-1)^i}
  \sum_{k=0}^i\frac{(-1)^k\binom{i}{k}}
  {\min\{d^{2(n+k-i)},d^n\}}.
\end{align}
For example, for the existence of the $4$-qubit AME state, the eigenvalues of 
the matrix
$\Phi_{AB}$ are
\begin{equation}
  (p_0,p_1,p_2,p_3,p_4)=\left(\frac{5}{864},0,\frac{1}{96},0,-\frac{1}{32}\right).
\end{equation}
The last negative eigenvalue implies that no $\AME(4,2)$ state exists.

\section{Multi-party extension: primal problem}
\label{app:SEGen}

\begin{figure}
  \centering
  \includegraphics[width=.65\textwidth]{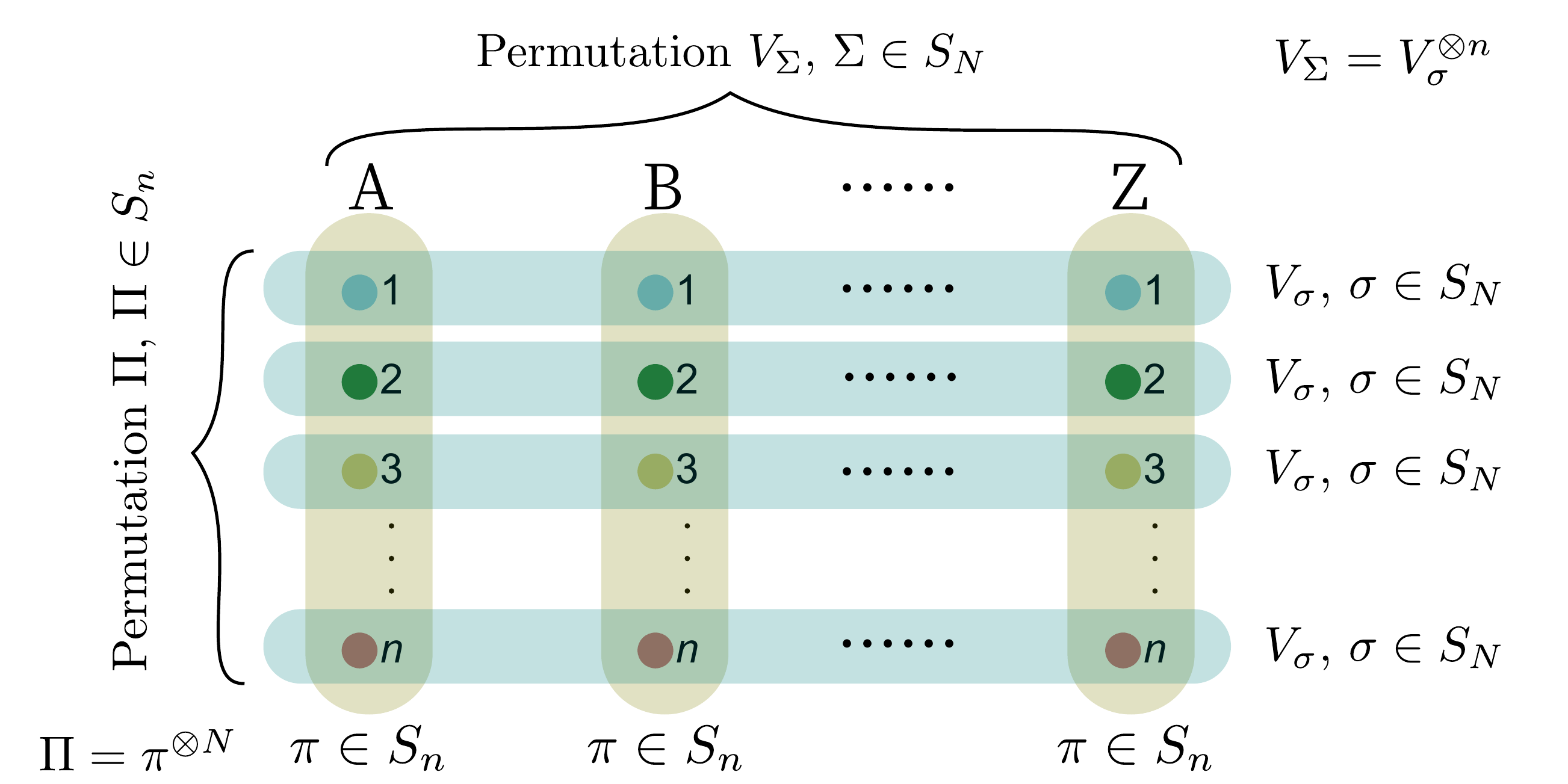}
  \caption{If the marginal problem has a solution $\ket{\varphi}$, then
  there are multi-party extensions $\Phi_{AB\cdots Z}$ for any number of 
copies, obeying some semidefinite constraints.}
  \label{fig:NcopyExt}
\end{figure}

We are going to analyze and simplify the hierarchy of SDPs stated in 
Theorem~\ref{thm:marginalPureHierarchy} for the case of the existence of AME 
states,
\begin{equation}
  \begin{aligned}
    &\findover\quad && \Phi_{AB\cdots Z}\\
    &\subto && P_N^+\Phi_{AB\cdots Z}P_N^+
    =\Phi_{AB\cdots Z},\\
    &       && \Phi_{AB\cdots Z}\ge 0,
    ~\Tr(\Phi_{AB\cdots Z})=1,\\
    &       && \Tr_{A_\bI}(\Phi_{AB\cdots Z})=\frac{\I_{d^r}}{d^r}\otimes
    \Tr(\Phi_{B\cdots Z})\FA I\in\mathcal{I}_r.\\
  \end{aligned}
  \label{eq:marginalFeasibilitySN}
\end{equation}
Similar to the two-party case, we can view the $N$-party state $\Phi_{AB\cdots 
Z}$ as $\Phi_{12\dots n}$, where $i$ labels the subsystems $A_iB_i\cdots Z_i$.  
The permutations on $A_iB_i\cdots Z_i$ are denoted with subscripts $ab\cdots 
z$.  For example, $V_{AB}$ and $V_{ABC}$ can be written as $V^{\otimes n}_{ab}$ 
and $V^{\otimes n}_{abc}$, respectively, where $V_{ab}$ are the permutations 
$A_i\leftrightarrow B_i$ and $V_{abc}$ are the permutations $A_i\rightarrow 
B_i\rightarrow C_i\rightarrow A_i$. Generally, we use $\sigma$ and $\Sigma$ to 
denote the permutations on $ab\cdots z$ and $AB\cdots Z$, respectively, and in 
addition $V_\Sigma=V_\sigma^{\otimes n}$.

Again, as the set of $\AME(n,d)$ is invariant under local unitaries and 
permutations on the $n$ particles, we can assume that $\Phi_{AB\cdots Z}$ is 
symmetric under the following operations,
\begin{align}
  \label{eq:groupExt1}
  &[U_1\otimes\dots\otimes U_n]^{\otimes N}\FA U_i\in SU(d),\\
  \label{eq:groupExt2}
  &\pi^{\otimes N}\FA \pi\in S_n.
\end{align}
Note that $\pi\in S_n$ denotes a permutation on $12\cdots n$ (vertical 
permutation in Fig.~\ref{fig:NcopyExt}), while $\sigma\in S_N$ in the previous 
paragraph  denotes a permutation on $ab\cdots z$ (horizontal permutation in 
Fig.~\ref{fig:NcopyExt}).  According to Schur-Weyl 
duality~\cite{Fulton.Harris1991}, any operator $\Phi$ such that $[\Phi, 
U^{\otimes N}]=0$ must have the form
\begin{equation}
  \Phi=\sum_{\sigma}x_\sigma V_\sigma.
  \label{eq:ShurWeyl}
\end{equation}
Thus, the $[U_1\otimes\dots\otimes U_n]^{\otimes N}$ symmetry implies that
\begin{equation}
    \Phi_{AB\cdots Z}
    =\sum_{\sigma_1\sigma_2\cdots\sigma_n}x_{\sigma_1\sigma_2\cdots\sigma_n}
    V_{\sigma_1}\otimes V_{\sigma_2}\otimes\cdots\otimes V_{\sigma_n}.
  \label{eq:WernerGeneral}
\end{equation}
The number of parameters can be further reduced by taking advantage of the 
vertical permutation symmetry $\{\Pi=\pi^{\otimes N}\mid\pi\in S_n\}$, i.e.,
\begin{equation}
  x_{\sigma_1\sigma_2\cdots\sigma_n}=x_{\sigma'_1\sigma'_2\cdots\sigma'_n}
  \label{eq:coSn}
\end{equation}
when $\{\sigma_1,\sigma_2,\cdots,\sigma_n\}$ and 
$\{\sigma'_1,\sigma'_2,\cdots,\sigma'_n\}$ are the same multiset (set that 
allows repeated elements).

We are now ready to express the constraints in 
\AEq~\eqref{eq:marginalFeasibilitySN} in terms of the variables 
$x_{\sigma_1\sigma_2\cdots\sigma_n}$ in \AEq~\eqref{eq:WernerGeneral}.
Naively plugging \AEq~\eqref{eq:WernerGeneral} into 
\AEq~\eqref{eq:marginalFeasibilitySN} results in relations between large 
matrices; however the symmetry of the problem allows one to also simplify these 
constraints.

Notice that the partial trace operation can also be expressed under the basis 
$\{V_\sigma\mid\sigma\in S_N\}$. For example,
\begin{equation}
  \begin{aligned}
    \Tr_c(\I)\otimes\I_c&=d\I,&
    \Tr_c(V_{ab})\otimes\I_c&=dV_{ab},&
    \Tr_c(V_{ac})\otimes\I_c&=\I,\\
    \Tr_c(V_{bc})\otimes\I_c&=\I,&
    \Tr_c(V_{abc})\otimes\I_c&=V_{ab},&
    \Tr_c(V_{cba})\otimes\I_c&=V_{ab},
  \end{aligned}
  \label{eq:ptMatrix}
\end{equation}
where all $V_\sigma$ are operators on $abc$ and we perform $\otimes\I_c$ to 
ensure that the operator stays within the original space.  Similarly, we can 
implement the trace operation. In this way, the equality constraints regarding 
the marginals in \AEq~\eqref{eq:marginalFeasibilitySN} can be written in terms 
of the basis operators $V_{\sigma_1}\otimes V_{\sigma_2}\otimes\cdots\otimes 
V_{\sigma_n}$ without referring to explicit matrix elements. Also, the 
symmetric projection $P_N^+$ takes the form
\begin{equation}
  P_N^+=\frac{1}{N!}\sum_{\sigma\in S_N}V_\sigma^{\otimes n}.
  \label{eq:symmetricSubspace}
\end{equation}
Therefore the equality $P_N^+\Phi_{AB\cdots Z}P_N^+
=\Phi_{AB\cdots Z}$ can also be expressed in terms of basis operators 
$V_{\sigma_1}\otimes V_{\sigma_2}\otimes\cdots\otimes V_{\sigma_n}$.

Let us now consider the positivity constraint $\Phi_{AB\cdots Z}\ge 0$.  Here, 
the crucial observation is that $\Phi_{AB\cdots Z}$ is simply a linear 
combination of the basic matrices $V_{\sigma_1}\otimes 
V_{\sigma_2}\otimes\cdots\otimes V_{\sigma_n}$. The matrices $V_{\sigma_i}$ in 
fact form a so-called (unitary linear) representation of the  group 
$S_N$~\cite{Fulton.Harris1991}. By the general theory of linear representations 
of groups, there is an orthogonal basis such that all of these matrices are 
block-diagonalized.  Moreover, the possible blocks that appear in the 
block-diagonal form of these matrices are also completely specified by the 
group, known as the unitary irreducible representations of the group. In this 
way, the positivity constraint on $\Phi_{AB\cdots Z}\ge 0$ is reduced to the 
positivity of each of the different irreducible blocks. 

For the symmetric group $S_N$, the irreducible representations are conveniently 
labeled by the partitions of $N$. A partition $\lambda$ of length $k= 
\abs{\lambda}$ is a tuple of positive integer numbers $\lambda 
= (N_1,N_2,\ldots,N_{k})$ such that $N_1 \ge N_2 \ge \cdots \ge N_{k}$ and $N_1 
+ N_2 + \cdots + N_{k} = N$.
We denote the set of all partitions by $\Lambda_N$.  For each partition 
$\lambda$, there is an associated unitary irreducible representation $M_\lambda$, 
that is, the set of unitary matrices $M_\lambda(\sigma)$ for $\sigma \in S_N$.  
Concretely, by choosing a suitable orthonormal basis (independent of $\sigma$), 
all $V_\sigma$ can be written as
\begin{equation}
  V_\sigma=\bigoplus_{\lambda}M_\lambda(\sigma)\otimes\I_{d_\lambda}
  \label{eq:SNdecom}
\end{equation}
where $M_\lambda(\sigma)$ correspond to the unitary irreducible representations 
and $d_\lambda$ are the corresponding multiplicities.
The matrix elements of $M_\lambda(\sigma)$ can also be constructed explicitly 
by taking advantage of the Young tableaux~\cite{Boerner1963}. For practical 
purposes, these matrices can be called from an appropriate computer algebra 
system such as GAP~\cite{GAP4}.  For the representation $V_{\sigma}$, it is 
also known that $M_\lambda(\sigma)$ is present ($d_\lambda\ne 0$) in the 
block-diagonal form of $V_{\sigma}$ if and only if the length of $\lambda$ is 
smaller than the local dimension $\abs{\lambda} \le 
d$~\cite{Fulton.Harris1991}. We thus have the following observation.

\begin{observation}
  For $\Phi_{AB\cdots Z}$  in \AEq~\eqref{eq:WernerGeneral}, $\Phi_{AB\cdots Z} 
  \ge 0$ if and only if
  \begin{equation}
    \sum_{\sigma_1\sigma_2\cdots\sigma_n}x_{\sigma_1\sigma_2\cdots\sigma_n}
    M_{\lambda_1}(\sigma_1)\otimes
    M_{\lambda_2}(\sigma_2)\otimes\cdots\otimes
    M_{\lambda_n}(\sigma_n)\ge 0,
    \label{eq:irrepPositivity}
  \end{equation}
  for all $(\lambda_1,\lambda_2,\dots,\lambda_n)\in\Lambda_N^n$ such that 
  $\abs{\lambda_i} \le d$. In addition, as the state $\Phi_{AB\cdots Z}$ is 
  also permutation-invariant under $\Pi\in S_n$, we can restrict to the cases 
  where $\lambda_1\ge\lambda_2\ge\dots\ge\lambda_n$ with any predefined order 
  for the partitions.
  \label{obs:irrepPositivity}
\end{observation}

There is yet another way to parameterize the optimization problem, which additionally
incorporates the constraint $P_N^+\Phi_{AB\cdots Z}P_N^+
    =\Phi_{AB\cdots Z}$ more directly.

Let us recall from the above that $\Phi_{AB \cdots Z}$ as well as $P_N^+$ are 
linear combinations of operators of the form $V_{\sigma_1}\otimes 
V_{\sigma_2}\otimes\cdots\otimes V_{\sigma_n}$. Thus, by choosing a suitable 
basis such that $V_{\sigma_i}$ are all block-diagonal, both $\Phi_{AB \cdots 
Z}$ and $P_N^+$ are also block-diagonal. The possible blocks of $V_{\sigma}$ 
are labeled by partitions of the form $\lambda = (N_1,N_2,\ldots,N_k)$ with 
$k=\abs{\lambda} \le d$. Correspondingly, the possible blocks of  
$V_{\sigma_1}\otimes V_{\sigma_2}\otimes\cdots\otimes V_{\sigma_n}$ are labeled 
by a tuple of partitions $\vl=(\lambda_1,\lambda_2,\ldots,\lambda_n)$ with 
$\abs{\lambda_i} \le d$. Each of such blocks may appear multiple times, but 
because of \AEq~\eqref{eq:SNdecom}, this simply results in exactly the same 
blocks in $\Phi_{AB \cdots Z}$ as well as $P_N^+$.  Therefore, considering just 
one time of appearance of each block is sufficient.  Moreover, because of the 
symmetry of coefficients in the linear combination under vertical permutations 
as in \AEq~\eqref{eq:coSn}, only a single representative of the tuples of 
partitions that are different by 
a vertical permutation needs to be considered. Hence, we are left with analyzing the 
constraint $P_N^+\Phi_{AB\cdots Z}P_N^+=\Phi_{AB\cdots Z}$ within the blocks 
corresponding to  $\vl=(\lambda_1,\lambda_2,\ldots,\lambda_n)$.

More specifically, let $\mathcal{H}_{\lambda_i}$ denote the subspace 
corresponding to the blocks $\lambda_i$ of the operators $V_{\sigma_i}$. Then 
the subspace corresponding to the block 
$\vl=(\lambda_1,\lambda_2,\ldots,\lambda_n)$ of $V_{\sigma_1}\otimes 
V_{\sigma_2}\otimes\cdots\otimes V_{\sigma_n}$ is given by
\begin{equation}
  \mathcal{H}_{\vect{\lambda}}=\mathcal{H}_{\lambda_1}
  \otimes\mathcal{H}_{\lambda_2}\otimes\cdots\otimes\mathcal{H}_{\lambda_n}. 
  \label{eq:orthogonalSubspaces}
\end{equation}
In this subspace, the symmetric projection $P_N^+$ reads
\begin{equation}
  (P_N^+)^{\vl}=\frac{1}{N!}\sum_{\sigma\in S_N}M_{\lambda_1}(\sigma)\otimes 
  M_{\lambda_2}(\sigma)\cdots\otimes M_{\lambda_n}(\sigma).
  \label{eq:symmetric_projection_reduced}
\end{equation}
The constraint $P_N^+\Phi_{AB\cdots Z}P_N^+=\Phi_{AB\cdots Z}$ restricted to the 
subspace $\mathcal{H}_{\vl}$ means that the corresponding block of 
$\Phi_{AB\cdots Z}$, denoted as $\Phi_{AB\cdots Z}^{\vl}$, is supported only on 
the symmetric subspace defined by the projection $(P_N^+)^{\vl}$,
\begin{equation}
  \mathcal{K}_{\vl}=\mathrm{Image}\left[(P_N^+)^{\vl}\right].
  \label{eq:eigenspace}
\end{equation}
Thus, if one chooses a basis 
$\{\ket{\Psi^{\vl}_i}\}_{i=1}^{k_{\vl}}$, where $k_{\vl} = \dim 
(\mathcal{K}_{\vl})$, for this subspace $\mathcal{K}_{\vl}$, then the 
corresponding block of $\Phi_{AB\cdots Z}$ is of the form
\begin{equation}
  \Phi^{\vect{\lambda}}_{AB \cdots Z}
  =\sum_{i,j=1}^{k_{\vect{\lambda}}}X_{ij}^{\vect{\lambda}}
  \ket{\Psi^{\vect{\lambda}}_i}\bra{\Psi^{\vect{\lambda}}_j}.
  \label{eq:eigenoperators}
\end{equation}
In this way, $\Phi^{\vect{\lambda}}_{AB\cdots Z}$ is parameterized by the 
matrix $X^{\vl}$, and its positivity reduces to the positivity of $X^{\vl}$.

In short, let us summarize the procedure to implement the optimization problem.
First, enumerate all irreducible representations of $S_N$, i.e., all 
possible partitions $\lambda$. Then, select those partitions that have length 
$\abs{\lambda}$ no longer than $d$. Based on that, enumerate all tuples of 
partitions $\vl=(\lambda_1,\lambda_2,\ldots,\lambda_n)$ with $\abs{\lambda_i} 
\le d$. For each of those tuples $\vl$, compute the symmetric projection 
$(P_N^+)^{\vl}$ by \AEq~\eqref{eq:symmetric_projection_reduced} and 
select a basis for $\mathcal{K}_{\vl} = \mathrm{Image} (P_N^+)^{\vl}$.
Finally, for each partition tuple $\vl$, consider the associated positive 
semidefinite Hermitian matrix variable $X^{\vl}$ and write down the constraints 
corresponding to the condition on the marginals in 
\AEq~\eqref{eq:marginalFeasibilitySN} to complete the SDP.

In addition, we provide some more details for the construction of the basis of 
$\mathcal{K}_{\vl}$. For readers who are familiar with the representation 
theory of groups, there is a simple characterization of $\mathcal{K}_{\vl}$ that 
helps carrying out the practical implementation. In the language of 
representation theory, $\mathcal{H}_{\lambda_i}$ is an irreducible 
representation of $S_N$, while $\mathcal{H}_{\vl}$ is an irreducible 
representation of $(S_N)^n$. This space is also a representation of $S_N$ via 
the diagonal embedding into $(S_N)^n$, which maps $\sigma \in S_N$  to 
$(\sigma,\sigma,\ldots,\sigma) \in (S_N)^n$. As a representation of $S_N$, 
$\mathcal{H}_{\vl}$ contains a subrepresentation $\mathcal{K}_{\vl}$ on which 
$S_N$ acts trivially (this is technically known as the isotropic component of 
the trivial representation).  Methods of representation theory then allow for 
detailed characterization of $\mathcal{K}_{\vl}$. In particular, one obtains the 
dimension of $\mathcal{K}_{\vl}$ as~\cite{Fulton.Harris1991}
\begin{equation}
k_{\vl} = \frac{1}{N!}\sum_{\sigma \in S_N} \prod_{i=1}^{n} \Tr(M_{\lambda_i} 
(\sigma)).  \end{equation}
The symmetric projection $(P_N^+)^{\vl}$ in 
\AEq~\eqref{eq:symmetric_projection_reduced} is in fact also known as the 
twirling operator: it maps a vector of $\mathcal{H}_{\vl}$ to its average under 
the action of the group $S_N$. A basis of this space can be found by applying 
the twirling operation  $(P_N^+)^{\vl}$ to a set of $k^{\vl}$ random vectors in 
$\mathcal{H}_{\vl}$; if the resulted vectors are linearly independent, they 
form a basis of $\mathcal{K}_{\lambda}$, else one can start over with another 
random set of vectors. As an alternative method, 
\AEqs~(\ref{eq:symmetric_projection_reduced},\,\ref{eq:eigenspace}) imply that 
$\mathcal{K}_{\vl}$ is the common unit eigenspace of 
$M_{\lambda_1}(\sigma)\otimes M_{\lambda_2}(\sigma)\otimes\cdots\otimes 
M_{\lambda_n}(\sigma)$ for all $\sigma\in S_N$. As all eigenvalues of 
$M_{\lambda_i}(\sigma)$ are always in the unit circle, a basis of 
$\mathcal{K}_{\vl}$ can also constructed from calculating the kernel of
\begin{equation}
  M_{\lambda_1}(\sigma_s)\otimes M_{\lambda_2}(\sigma_s)\otimes\dots\otimes 
  M_{\lambda_n}(\sigma_s)+
  M_{\lambda_1}(\sigma_c)\otimes M_{\lambda_2}(\sigma_c)\otimes\dots\otimes 
  M_{\lambda_n}(\sigma_c)-
  2\I,
  \label{eq:kernel}
\end{equation}
where $\sigma_s=(ab)$ and $\sigma_c=(ab\cdots z)$ form a set of generators of 
$S_N$.

As another technical remark, working with unitary representation requires 
computation with cyclotomic numbers, which is often slow. Therefore, one may 
adjust the procedure by implementing intermediate computations in non-unitary 
representations (or equivalently, working in non-orthogonal bases) where matrix 
elements (of the representations of symmetric groups) are all rationals.

\section{Multi-party extension: dual problem and entanglement witness}
\label{app:SEAME}

Specifically for the existence problem of AME states, as $\Phi_{AB}$ is 
uniquely determined, one can easily verify that the following equation is 
a relaxed but still complete hierarchy of 
Theorem~\ref{thm:marginalPureHierarchy},
\begin{equation}
  \begin{aligned}
    &\findover\quad && \Phi_{ABC\cdots Z}\\
    &\subto && \Tr_{C\cdots Z}(P_N^+\Phi_{ABC\cdots Z}P_N^+)=\Phi_{AB},\\
    &       && P_N^+\Phi_{ABC\cdots Z}P_N^+\ge 0,
  \end{aligned}
  \label{eq:symmetricExtension}
\end{equation}
where $\Phi_{AB}$ is the unique quantum state given by Theorem~\ref{thm:AME}.  
Alternatively, we can write the objective function in 
\AEq~\eqref{eq:symmetricExtension} as $\max_{\Phi_{ABC\cdots Z}}\{0\}$, such 
that the dual problem reads
\begin{equation}
  \begin{aligned}
    &\minover[W_{AB}]\quad && \Tr(W_{AB}\Phi_{AB})\\
    &\subto && P_N^+W_{AB}\otimes\I_{C\cdots Z}P_N^+\ge 0,
  \end{aligned}
  \label{eq:entanglementWitness}
\end{equation}
where $W_{AB}$ is Hermitian. One can easily verify that strong duality 
holds from Slater's condition~\cite{Boyd.Vandenberghe2004}
with positivity considered on the symmetric subspace, which means the 
problem in \AEq~\eqref{eq:symmetricExtension} is feasible if and only if the 
solution of the dual problem in \AEq~\eqref{eq:entanglementWitness} equals zero.  
Thus, if $\Tr(W_{AB}\Phi_{AB})<0$, we know that $\Phi_{AB}$ is 
entangled and the corresponding AME state does not exist from 
Theorem~\ref{thm:AME}. Notice that numerically determining the negativity of 
the dual problem in \AEq~\eqref{eq:entanglementWitness} is less sensitive to 
small numerical errors, and hence, more stable than solving the primal 
feasibility problem in \AEq~\eqref{eq:symmetricExtension}. Moreover, the 
physical meaning of $W_{AB}$ is also clear: a feasible point $W_{AB}$ of 
\AEq~\eqref{eq:entanglementWitness} with a negative objective value provides an 
entanglement witness for $\Phi_{AB}$ in the symmetric subspace 
$P_2^+=\frac{1}{2}(\I_{AB}+V_{AB})$.  Indeed, because the set of separable 
states in $P_2^+$ is given by $\conv\{\ket{\psi} 
\bra{\psi}\otimes\ket{\psi}\bra{\psi}\}$, the constraint in 
\AEq~\eqref{eq:entanglementWitness} implies that
\begin{equation}
  \bra{\psi}\bra{\psi}W_{AB}\ket{\psi}\ket{\psi}=
  \bra{\psi}^{\otimes N}P_N^+W_{AB}\otimes\I_{C\cdots Z}P_N^+
  \ket{\psi}^{\otimes N}\ge 0.
  \label{eq:asWitness}
\end{equation}

The analysis of the symmetry and parametrization of the dual problem 
\AEq~\eqref{eq:entanglementWitness} is similar to that for the primal problem 
as discussed in \Appendix~\ref{app:SEGen}; in fact, it is more straightforward 
for the dual problem. For $g\in G$ defined in 
\AEqs~(\ref{eq:group1},\,\ref{eq:group2}), we have
\begin{equation}
  g\Phi_{AB}g^\dagger=\Phi_{AB},\quad gP_N^+g^\dagger=P_N^+.
  \label{eq:symmetryDual}
\end{equation}
In addition, we know that $\Phi_{AB}$ and $P_N^+$ are also in the symmetric 
subspace $P_2^+$, i.e.,
\begin{equation}
  P_2^+\Phi_{AB}P_2^+=\Phi_{AB},\quad\left(P_2^+\otimes\I_{C\cdots 
  Z}\right)P_N^+\left(P_2^+\otimes\I_{C\cdots Z}\right)=P_N^+.
  \label{eq:symmetryDualAdd}
\end{equation}
Thus, we can assume that $W_{AB}$ is invariant under $G$ and constrained to 
$P^+_2$, i.e.,
  \begin{equation}
    gW_{AB}g^\dagger=W_{AB} \FA g\in G, \quad P_2^+W_{AB}P_2^+=W_{AB}.
  \label{eq:entanglementWitnessSim}
\end{equation}

Similar to the analysis of \MEq~\eqref{eq:PhiSymm}, one can easily see that 
$gW_{AB}g^\dagger=W_{AB}$ for all $g \in G$ implying that
\begin{equation}
  W_{AB}=\sum_{l=0}^nw_l\cP\{V^{\otimes l}\otimes\I^{\otimes (n-l)}\},
  \label{eq:WSymm}
\end{equation}
where again $\cP$ denotes the sum over all permutations of the tensor 
product under its argument.
Furthermore, $P_+W_{AB}P_+=W_{AB}$ implies that
\begin{equation}
  w_l=w_{n-l} \FA l=0,1,\dots,n-r-1.
  \label{eq:symmetric}
\end{equation}
Hence, the objective function $\Tr(W_{AB}\Phi_{AB})$ can be expressed as
\begin{equation}
  \Tr(W_{AB}\Phi_{AB})=\sum_{l=0}^na_lw_l,
  \label{eq:objSimDual}
\end{equation}
where
\begin{equation}
  a_l=\Tr(\cP\{V^{\otimes l}\otimes\I^{\otimes (n-l)}\}\Phi_{AB})
  =\frac{\binom{n}{l}}{\min\{d^l,d^{n-l}\}},
  \label{eq:al}
\end{equation}
from \AEq~\eqref{eq:ai}. 

To get some intuition about the variables $w_l$, let us consider the problem of 
the existence of AME(4,6). Here $n=4$ and hence, there 
are five variables $w_l$ in \AEq~\eqref{eq:WSymm}. Moreover, 
\AEq~\eqref{eq:symmetric} implies that only three of those variables 
are independent. Furthermore, one can notice that the dual problem in 
\AEq~\eqref{eq:entanglementWitness} is homogeneous, that is, the objective 
function is linear and the constraints are invariant under rescaling $W_{AB} 
\to t W_{AB}$ with $t > 0$. This allows one to impose that $w_0=0$ or $w_0=\pm 
1$, and one is then left with two independent variables.

The constraint $P_N^+W_{AB}\otimes\I_{C\cdots Z}P_N^+\ge 0$ can be expressed in 
terms of the variables $w_l$ in similarity to \Appendix~\ref{app:SEGen}.  Let 
us summarize the arguments once more for completeness. The fact that 
$W_{AB}\otimes\I_{C\cdots Z}$ and $P_N^+$ are both of linear combinations of 
$V_{\sigma_1}\otimes V_{\sigma_2}\otimes\cdots\otimes V_{\sigma_n}$ implies 
that they are block-diagonal when one chooses a basis such that the 
$V_{\sigma_i}$ are block-diagonal. Let $\mathcal{H}_{\lambda_i}$ denote the 
subspace corresponding to the block of $V_{\sigma_i}$ labeled by partition 
$\lambda_i$ with $\abs{\lambda_i} \le d$.  Then $\mathcal{H}_{\vl} 
= \mathcal{H}_{\lambda_1} \otimes \mathcal{H}_{\lambda_2} \otimes 
\mathcal{H}_{\lambda_n}$ denotes the subspace corresponding to a block of 
$V_{\sigma_1}\otimes V_{\sigma_2}\otimes\cdots \otimes V_{\sigma_n}$ labeled by 
a tuple of partitions $\vl=(\lambda_1,\lambda_2,\ldots,\lambda_n)$.  Moreover, 
within this subspace, $(P_N^+)^{\vl}$ is a projection onto the symmetric 
subspace, which is typically low-rank. Let $\mathcal{K}_{\vl}$ denote the image 
of $(P_N^+)^{\vl}$ and $\{\ket{\Psi_i^{\vl}}\}_{i=1}^{k_{\vl}}$ denote a basis 
of $\mathcal{K}_{\vl}$.  One defines the matrix $Y^{\vl}$ as
  \begin{equation}
    Y^{\vl}_{ij} = \bra{\Psi^{\vl}_i} P_N^+W_{AB}\otimes\I_{C\cdots Z}P_N^+ 
    \ket{\Psi^{\vl}_j}.
  \end{equation}
Notice that in computing these matrix elements, we only need the blocks of 
$P_N^+$ and $W_{AB}\otimes\I_{C\cdots Z}$ corresponding to partitions $\vl$.  
Then, $P_N^+W_{AB}\otimes\I_{C\cdots Z}P_N^+\ge 0$ is equivalent to $Y^{\vl} 
\ge 0$ for all tuples of partitions $\vl$ with $\abs{\lambda_i} \le d$.  
Moreover, since the problem is symmetric under vertical permutations, tuples of 
partitions $\vl$ that are different by a vertical permutation are considered 
just once.

As a final remark, we can consider the relaxations of the constraints in 
\AEq~\eqref{eq:entanglementWitness}. If the optimal value of a relaxed problem 
is non-negative, we conclude that the optimal value of 
\AEq~\eqref{eq:entanglementWitness} is also non-negative.  In particular, 
ignoring some tuples of partitions $\vl$ in the constraints $Y^{\vl} \ge 0$ 
corresponds to a relaxation of \AEq~\eqref{eq:entanglementWitness}. For example, 
one can consider only $\vl$ such that $(P_N^+)^{\vl}$ is rank-$1$ and obtain 
a linear program relaxation of \AEq~\eqref{eq:entanglementWitness}.

\section{Failed approaches to the AME problem}
\label{app:fail}

In this section, we discuss the approaches that we applied to investigate the 
separability of states which encode the existence of
AME states. For the interesting case of $\AME(4,6)$, however, none
of them delivers a solution to the problem.

\subsection{The state for the $\AME(4,6)$ problem}
Let us start by recalling the state presented already in Corollary~\ref{cor:AME46}.
The state acts on a $6^4 \times 6^4$ system, where Alice and Bob each own four 
six-dimensional systems. The state is given by
\begin{equation}
    \Phi_{AB}= \frac{1}{2\cdot 6^4} \left(
    \frac{P_+^{\otimes 4}}{343}
    +\frac{\cP\big\{P_+^{\otimes 2}
    \otimes P_-^{\otimes 2}\big\}}{315}
    +\frac{P_-^{\otimes 4}}{375} \right),
    \label{eq:Phi46A}
\end{equation}
where $P_\pm$ are the projectors onto the (anti-)symmetric subspace 
of the $6\times6$ systems. Here, the tensor product denotes the tensor product 
between the four $6\times6$ systems and $\mathcal{P}\{\cdot\}$ denotes a sum 
over all permutations of the four copies that give distinct terms; in this 
case, there are six different terms. Note that the state $\Phi_{AB}$ acts on 
the symmetric subspace only.

It is also useful to consider the partial transposition of this state.
Let $\ket{\phi^+}=(\sum_{k=0}^{5}\ket{kk})/\sqrt{6}$ be the maximally entangled 
state of two six-dimensional systems and define $P_\perp=\I-\ketbra{\phi^+}$ as 
the projector onto the corresponding orthogonal subspace.  Then, we have
\begin{equation}
    \Phi_{AB}^{T_B}= \frac{1}{6^4} \left(
    \ket{\phi^+}\bra{\phi^+}^{\otimes 4}
    +\frac{\cP\big\{\ket{\phi^+}\bra{\phi^+}\otimes
    P_\perp^{\otimes 3}\big\}}{35^2}
    +\frac{33 P_\perp^{\otimes 4}}{35^3} \right).
    \label{eq:Phi46TA}
\end{equation}
This time, the sum over all permutations contains four different terms.
Clearly, the separability of $\Phi_{AB}$ is equivalent to the separability
of $\Phi_{AB}^{T_B}$. To test whether or not these states are entangled the 
following approaches came to our mind:

\begin{itemize}

\item
The state $\Phi_{AB}^{T_B}$ has a similarity to the states discussed in 
Ref.~\cite{Huber.etal2018b}. There, a family of bound entangled
states with high Schmidt rank has been constructed. To do so, one
considers a bipartite system, where Alice's as well as Bob's system can be 
further split up into two subsystems, $A_1$ and $A_2$ as well as $B_1$ and 
$B_2$, respectively. Then, one investigates unnormalized states of the form
\begin{equation}
  Z = X_{A_1 B_1} \otimes (P_\perp)_{A_2 B_2}
  + Y_{A_1 B_1}\otimes\ketbra{\phi^+}_{A_2 B_2}.
  \label{eq:bestates}
\end{equation}
Under weak conditions on $X_{A_1 B_1}$ and $Y_{A_1 B_1}$ one can show
that $Z$ is a bipartite entangled state with a positive partial transpose. 
For instance, one may choose $X_{A_1 B_1}=(P_\perp)_{A_1 B_1}$
and $Y_{A_1 B_1} =(d_1-1)(d_2+1)\ketbra{\phi^+}_{A_1 B_1}$. Here, $d_1$ is the
dimension of $A_1$ and $B_1$ and $d_2$ the dimension of $A_2$ and $B_2$. For
the argument of Ref.~\cite{Huber.etal2018b} it is crucial that these dimensions
are different, typically one takes $d_2 \gg d_1$.

The entanglement proof for the states in Ref.~\cite{Huber.etal2018b} goes as 
follows: The map \begin{equation}
  \Lambda(\cdot) = \I\Tr(\cdot)-\frac{1}{k}\id(\cdot),
\end{equation}
is $k$-positive, where $\id(\cdot)$ denotes the identity map. That is, the 
output of $\id\otimes\Lambda$ is always positive on states with Schmidt rank 
$k$.  A non-positive output by applying this map to the $A_2 B_2$ part of 
states of the form in \AEq~(\ref{eq:bestates}), i.e., applying $\id_{A_1 B_1 
A_2}\otimes\Lambda_{B_2}$, would indicate that the state has a very high 
Schmidt rank in the systems $A_2 B_2$. The (low-dimensional) systems $A_1 B_1$ 
cannot significantly change the Schmidt rank, so the total state must be 
entangled. This idea can also be formalized by writing down explicit 
entanglement witnesses \cite{Huber.etal2018b}.

For the state $\Phi_{AB}^{T_B}$ one can apply similar tricks. For instance,
one can split the four subsystems of Alice and Bob in a one-vs-three partition
to achieve $d_2 \gg d_1$. In this particular case, however, the state is not
detected as entangled, the expectation value of the witness from Ref.~
\cite{Huber.etal2018b} vanishes. One may also consider further refined splits, 
as any six-dimensional system can be seen as a $(2\times3)$-system. For 
example, one can split the system such that $d_1= 2^4=16$ and $d_2= 3^4=81$.  
Still, we found no proof of entanglement for $\Phi_{AB}^{T_B}$, however, the 
expectation value for several of the resulting witnesses vanishes.

\item
Similar states as in Ref.~\cite{Huber.etal2018b} were also considered before in 
Ref.~\cite{Piani.Mora2007}. There, entanglement witnesses of the form
\begin{equation}
W=\ketbra{\psi_1}_{A_1 B_1} \otimes \I_{A_2 B_2}
-(1+\varepsilon)\ketbra{\psi_1}_{A_1 B_1} \otimes \ketbra{\psi_2}_{A_2 B_2}
\label{eq:pianiwitness}
\end{equation}
have been investigated. For the purpose of Ref.~\cite{Piani.Mora2007}, it was 
only relevant that for some $\varepsilon>0$ this operator is indeed positive on 
all separable states, and it was shown that this holds for nearly arbitrary 
$\ket{\psi_1}$ and $\ket{\psi_2}$.

For our purposes, we need to calculate the maximal $\varepsilon$
explicitly. If we assume that $\ket{\psi_1}$ and $\ket{\psi_2}$
are maximally entangled states in different dimensions, this can
be done as follows: First, we know that $W_k=k/d_2-\ketbra{\psi_2}$ is 
a Schmidt rank-$k$ witness. Second, if we consider a product state
$\ket{\eta}=\ket{\alpha}_{A_1 A_2} \otimes \ket{\beta}_{B_1 B_2}$,
the unnormalized pure state
\begin{equation}
\ketbra{\zeta}_{A_2 B_2} = \Tr_{A_1 B_1}\big[\ketbra{\eta}_{A_1 A_2 B_1 B_2} 
\ketbra{\psi_1}_{A_1 B_1}\big],
\end{equation}
has at most Schmidt rank $d_1$. Combining these observations, we find that $W$ 
in \AEq~(\ref{eq:pianiwitness}) is an entanglement witness if
\begin{equation}
\varepsilon \leq \frac{d_2}{d_1}-1.
\end{equation}
For instance, taking the state $\Phi_{AB}^{T_B}$ as well as $d_1=2$
and $d_2= 6^3 \times 3$, one obtains $\varepsilon = 323$.
Still, we find $\Tr(W\Phi_{AB}^{T_B})=0$ and no entanglement is detected.

\item
As described in \Appendix~\ref{app:SEGen}, we also tested whether or not there 
exists a symmetric extension for the state $\Phi_{AB}$ making use of the 
symmetries to reduce the number of parameters substantially. However, for large 
extensions, computing the bases for $\mathcal{K}_{\vl}$ in 
\AEq~\eqref{eq:eigenspace} as well as rephrasing the constraints in terms of 
the variables in \AEq~\eqref{eq:eigenoperators} takes a considerable amount of 
time.  Moreover, precision issues pose a major challenge due to coefficients 
being of different order of magnitude.

One possible relaxation that simplifies the computation is to consider the 
second last constraint in the SDP in \AEqs~\eqref{eq:marginalFeasibilitySN} only 
for some marginal of the extension.  The largest extension we computed reliably 
is $N=5$ while restricting the second last constraint to 
$\Phi_{ABC}=\Tr_{DE}(\Phi_{ABCDE})$. Furthermore, we computed a PPT-extension 
for $N=3$ utilizing the basis from Ref.~\cite{Eggeling.Werner2001}. Both of 
these extensions exist up to numerical precision.

\item
We implemented the dual problem in \AEq~\eqref{eq:entanglementWitness} 
exploiting its symmetry as discussed in \Appendix~\ref{app:SEAME}. Using the 
linear program relaxation of the problem by means of retaining only partitions 
$\vl$ such that the symmetric projection $(P^+_N)^{\vl}$ is rank-$1$ as 
discussed there, we can show that the optimal values are non-negative up to 
$N=7$. Thus the hierarchy fails to indicate the possible entanglement of 
$\Phi_{AB}$ up to $N=7$.

\item
A final idea could be to start with the symmetric state $\Phi_{AB}$ and use the 
following strategy to prove that the state is entangled: For a multiparticle
symmetric state it is known that it is either fully separable or genuine multipartite
entangled. This implies that if a multiparticle symmetric state is entangled for
one bipartition, it must be entangled for all bipartitions. Hence, proving entanglement
for one bipartition can be used to show entanglement for another bipartition, even
if the state has a positive partial transpose for the latter bipartition. This 
trick has been exploited to find symmetric bound entangled states 
\cite{Toth.Guehne2009}.

For the state $\Phi_{AB}$ one would need to find an embedding in a multiparticle 
system, where $\Phi_{AB}$ corresponds to some bipartition. This, however, is not straightforward, 
as the embedding idea from Ref.~\cite{Toth.Guehne2009} does not work for bipartite symmetric states
with maximal rank. 
\end{itemize}

\subsection{The state for the $\AME(7,2)$ problem}

For training purposes, it may be useful to consider a state where 
the separability properties are known. The following state originates from the seven-qubit
AME problem, where no AME state exists \cite{Huber.etal2017}. It is, however, 
not easy to see the entanglement of the corresponding state directly, 
and finding a criterion might also help to decide whether or not there is an $\AME(4,6)$ state. 

The state acts on a $2^7 \times 2^7$ system, where Alice and Bob each own
seven qubits:
\begin{equation}
  \Phi_{AB}=
  \frac{113}{1119744}P_+^{\otimes 7}
  +\frac{17}{124416}\mathcal{P}
  \left\{P_+^{\otimes 5}\otimes P_-^{\otimes 2}\right\}
  +\frac{1}{13824}\mathcal{P}
  \left\{P_+^{\otimes 3}\otimes P_-^{\otimes 4}\right\}
  +\frac{1}{1536}\mathcal{P}
  \left\{P_+^{\otimes 1}\otimes P_-^{\otimes 6}\right\},
  \label{eq:AME72}
\end{equation}
where $P_\pm$ are the projectors onto the (anti-)symmetric subspace of the 
$2\times2$ systems.

For the partial transposition, let $\ket{\phi^+}=(\ket{00}+\ket{11})/\sqrt{2}$ 
be the two-qubit Bell state, and $P_\perp = \I - \ketbra{\phi^+}$ the projector 
onto the corresponding orthogonal subspace. Then,
\begin{equation}
  \begin{aligned}
    \Phi_{AB}^{T_B}&=
    \frac{1}{128}\ketbra{\phi^+}^{\otimes 7}
    +\frac{1}{10368}\mathcal{P}
    \left\{\ketbra{\phi^+}^{\otimes 3}\otimes P_\perp^{\otimes 4}\right\}
    +\frac{1}{15552}\mathcal{P}
    \left\{\ketbra{\phi^+}^{\otimes 2}\otimes P_\perp^{\otimes 5}\right\}\\
    &+\frac{1}{23328}\mathcal{P}
    \left\{\ketbra{\phi^+}\otimes P_\perp^{\otimes 6}\right\}
    +\frac{11}{139968}P_\perp^{\otimes 7}.
    \label{eq:AME72PT}
  \end{aligned}
\end{equation}
As no $\AME(7,2)$ state exists, the states $\Phi_{AB}$ and $\Phi_{AB}^{T_B}$ in 
\AEqs~(\ref{eq:AME72},\,\ref{eq:AME72PT}) are entangled, but we are not aware 
of any operational entanglement criterion detecting them.

\twocolumngrid
\bibliography{QuantumInf}

\end{document}